\documentclass[a4paper,pra,groupedaddress,twocolumn,nofootinbib,notitlepage,accepted=2019-05-16,nopdfoutputerror]{quantumarticle}

\pdfoutput=1
\usepackage[english]{babel}
\usepackage[T1]{fontenc}
\usepackage[utf8]{inputenc}
\usepackage{hyperref}

\usepackage{ytableau}
\usepackage{enumerate}

\usepackage{array}
\usepackage{booktabs}
\usepackage{units}
\usepackage{amsmath}
\usepackage{amsthm}
\usepackage{amssymb}
\usepackage{graphicx}

\makeatletter

\providecommand{\tabularnewline}{\\}

\theoremstyle{definition}
\newtheorem*{defn*}{\protect\definitionname}
\theoremstyle{plain}
\newtheorem*{lyxalgorithm*}{\protect\algorithmname}
\theoremstyle{plain}
\newtheorem{thm}{\protect\theoremname}
\theoremstyle{plain}
\newtheorem{prop}[thm]{\protect\propositionname}

\makeatother

\providecommand{\algorithmname}{Algorithm}
\providecommand{\definitionname}{Definition}
\providecommand{\propositionname}{Proposition}
\providecommand{\theoremname}{Theorem}

\begin{document}
\global\long\def\bra{\langle}%
\global\long\def\ket{\rangle}%
\global\long\def\half{\frac{1}{2}}%
\global\long\def\third{\frac{1}{3}}%
\global\long\def\dx{\frac{\partial}{\partial x}}%
\global\long\def\dxd{\frac{\partial}{\partial\dot{x}}}%
\global\long\def\thrha{\frac{3}{2}}%
\global\long\def\sq{\sqrt{2}}%
\global\long\def\sqinv{\frac{1}{\sqrt{2}}}%
\global\long\def\up{\uparrow}%
\global\long\def\do{\downarrow}%
\global\long\def\p{\partial}%
\global\long\def\dqi{\frac{\partial}{\partial q_{i}}}%
\global\long\def\dqid{\frac{\partial}{\partial\dot{q}_{i}}}%
\global\long\def\a{\alpha}%
\global\long\def\b{\beta}%
\global\long\def\g{\gamma}%
\global\long\def\c{\chi}%
\global\long\def\d{\delta}%
\global\long\def\o{\omega}%
\global\long\def\m{\mu}%
\global\long\def\n{\nu}%
\global\long\def\z{\zeta}%
\global\long\def\l{\lambda}%
\global\long\def\e{\epsilon}%
\global\long\def\x{\chi}%
\global\long\def\r{\rho}%
\global\long\def\t{\theta}%
\global\long\def\G{\Gamma}%
\global\long\def\D{\mathcal{D}}%
\global\long\def\O{\mathcal{O}}%
\global\long\def\L{\mathcal{L}}%
\global\long\def\T{\mathcal{T}}%
\global\long\def\db{\dbar}%
\global\long\def\dg{\dagger}%
\global\long\def\s{\sigma}%
\global\long\def\ph{\hat{n}}%
\global\long\def\RR{\mathcal{R}}%
\global\long\def\phm{\hat{m}}%
\global\long\def\bm{b_{-}^{\dg}b_{-}}%
\global\long\def\bp{b_{+}^{\dg}b_{+}}%
\global\long\def\lp{\biggl(}%
\global\long\def\rp{\biggr)}%
\global\long\def\pr{\prime}%
\global\long\def\rin{\rho_{\textsf{in}}}%
\global\long\def\hs{\mathsf{H}_{\text{out}}}%
\global\long\def\rss{\rho_{\mathsf{ss}}}%
\global\long\def\pis{\Pi_{\text{out}}}%
\global\long\def\uout{U_{\text{out}}}%
\global\long\def\gout{G_{\text{out}}}%
\global\long\def\ph{\hat{n}}%
\global\long\def\pout{\Pi_{\textsf{ss}}}%
\global\long\def\mout{\mathsf{L_{ss}}}%
\global\long\def\LL{\mathcal{L}}%
\global\long\def\hi{F_{l}}%
\global\long\def\Al{\mathcal{A}_{l}}%
\global\long\def\pn{\pi_{0}}%
\global\long\def\as{\alpha}%
\global\long\def\aaa{a}%
\global\long\def\dag{\dagger}%
\global\long\def\ma{\left|\a\right|}%
\global\long\def\A{\mathcal{A}}%
\global\long\def\tr{\text{Tr}}%
\global\long\def\jp{J_{\m}^{\perp}}%
\global\long\def\I{\mathcal{I}}%
\global\long\def\H{\mathcal{H}}%
\global\long\def\PP{\mathcal{P}}%
\global\long\def\R{\mathcal{P}}%
\global\long\def\pl{\mathcal{P}_{\L}}%
\global\long\def\ot{\otimes}%
\global\long\def\bath{\varrho_{\textsf{aux}}}%
\global\long\def\ns{\text{NS}}%
\global\long\def\pa{\partial^{\prime}}%
\global\long\def\no{n_{\textsf{aux}}}%
\global\long\def\pns{\pi_{\text{NS}}}%
\global\long\def\pa{\pi_{\textsf{aux}}}%
\global\long\def\pp{P}%
\global\long\def\k{\kappa}%
\global\long\def\Q{\mathcal{Q}}%
\global\long\def\oph{\text{Op}(\mathsf{H})}%
\global\long\def\h{\mathsf{H}}%
\global\long\def\bb{\langle\!\langle}%
\global\long\def\kk{\rangle\!\rangle}%
\global\long\def\da{d_{\textsf{aux}}}%
\global\long\def\hns{\mathsf{H}_{\textsf{NS}}}%
\global\long\def\ha{\mathsf{H}_{\textsf{aux}}}%
\global\long\def\KK{\mathcal{H}^{\pr}}%
\global\long\def\pert{H^{\pr}}%
\global\long\def\kns{H_{\text{\textsf{NS}}}^{\pr}}%
\global\long\def\ka{H_{\textsf{aux}}^{\pr}}%
\global\long\def\W{\mathcal{W}}%
\global\long\def\kkk#1{\left.\left|#1\right\rangle \!\!\!\right\rangle }%
\global\long\def\ld{\dot{\l}}%
\global\long\def\GG{\mathcal{G}}%
\global\long\def\QQ{\mathcal{Q}_{\textsf{ss}}}%
\global\long\def\id{\mathcal{I}}%
\global\long\def\U{\mathcal{U}}%
\global\long\def\ppp{\mathcal{P}_{\!\!\A}}%
\global\long\def\E{\mathcal{E}}%
\global\long\def\J{\mathcal{J}}%
\global\long\def\dgt{\ddagger}%
\global\long\def\la{\varDelta}%
\global\long\def\ps{\mathcal{P_{\E}}}%
\global\long\def\M{\mathcal{M}}%
\global\long\def\B{\mathcal{B}}%
\global\long\def\dfunc{\mathcal{J}}%
\global\long\def\jd{J_{\textnormal{\textsf{dgn}}}}%
\global\long\def\jn{J_{\textnormal{\textsf{nil}}}}%
\global\long\def\Z{\mathbb{Z}}%
\global\long\def\qq{Q}%
\global\long\def\st{\varPsi}%
\global\long\def\St{\varPsi}%
\global\long\def\P{\Phi}%
\global\long\def\Tr{\textsc{Tr}}%
\global\long\def\fe{\textnormal{fix}({\cal E})}%
\global\long\def\rout{\r_{\infty}}%
\global\long\def\fed{\textnormal{fix}({\cal E^{\dgt}})}%
\global\long\def\bo{\widetilde{B}}%
\global\long\def\stt{\widetilde{\St}}%

\DeclareRobustCommand{\ul}{{\raisebox{2pt}{\ytableaushort{ {*(black)} {} , {} {} }}}} 
\DeclareRobustCommand{\ur}{{\raisebox{2pt}{\ytableaushort{ {} {*(black)} , {} {} }}}} 
\DeclareRobustCommand{\ll}{{\raisebox{2pt}{\ytableaushort{ {} {} , {*(black)} {} }}}} 
\DeclareRobustCommand{\lr}{{\raisebox{2pt}{\ytableaushort{ {} {} , {} {*(black)} }}}} 
\DeclareRobustCommand{\di}{{\raisebox{2pt}{\ytableaushort{ {*(black)} {} , {} {*(black)} }}}} 
\DeclareRobustCommand{\of}{{\raisebox{2pt}{\ytableaushort{ {} {*(black)} , {*(black)} {} }}}}
\DeclareRobustCommand{\thr}{{\raisebox{2pt}{\ytableaushort{ {*(black)} {*(black)} , {} {*(black)} }}}}
\DeclareRobustCommand{\thu}{{\raisebox{2pt}{\ytableaushort{ {*(black)} {*(black)} , {*(black)} {} }}}}
\DeclareRobustCommand{\tho}{{\raisebox{2pt}{\ytableaushort{ {} {*(black)} , {*(black)} {*(black)} }}}}
\DeclareRobustCommand{\thd}{{\raisebox{2pt}{\ytableaushort{ {*(black)} {*(black)} , {} {*(black)} }}}}
\DeclareRobustCommand{\emp}{{\raisebox{2pt}{\ytableaushort{ {} {} , {} {} }}}}
\ytableausetup{boxsize = 2pt}
\newcommand{\ulbig}{\ytableausetup{boxsize = 3pt}
\raisebox{3pt}{\ytableaushort{ {*(black)} {} , {} {} }}
\ytableausetup{boxsize = 2pt}}
\newcommand{\ofbig}{\ytableausetup{boxsize = 3pt}
\raisebox{3pt}{\ytableaushort{ {} {*(black)} , {*(black)} {} }}
\ytableausetup{boxsize = 2pt}}
\newcommand{\urbig}{\ytableausetup{boxsize = 3pt}
\raisebox{3pt}{\ytableaushort{ {} {*(black)} , {} {} }}
\ytableausetup{boxsize = 2pt}}
\newcommand{\llbig}{\ytableausetup{boxsize = 3pt}
\raisebox{3pt}{\ytableaushort{ {} {} , {*(black)} {} }}
\ytableausetup{boxsize = 2pt}}
\newcommand{\lrbig}{\ytableausetup{boxsize = 3pt}
\raisebox{3pt}{\ytableaushort{ {} {} , {} {*(black)} }}
\ytableausetup{boxsize = 2pt}}
\newcommand{\thubig}{\ytableausetup{boxsize = 3pt}
\raisebox{3pt}{\ytableaushort{ {*(black)} {*(black)} , {*(black)} {} }}
\ytableausetup{boxsize = 2pt}}
\newcommand{\empbig}{\ytableausetup{boxsize = 3pt}
\raisebox{3pt}{\ytableaushort{ {} {} , {} {} }}
\ytableausetup{boxsize = 2pt}}
\newcommand{\dibig}{\ytableausetup{boxsize = 3pt}
\raisebox{3pt}{\ytableaushort{ {*(black)} {} , {} {*(black)} }}
\ytableausetup{boxsize = 2pt}}
\newenvironment{numtheorem}[1]{\begin{trivlist}\item[\hskip \labelsep {\bf #1}]\it}{\end{trivlist}}

\title{Asymptotics of quantum channels: conserved quantities, an adiabatic
limit, and matrix product states}
\author{Victor~V.~Albert}
\affiliation{Walter Burke Institute for Theoretical Physics and Institute for Quantum
Information and Matter, California Institute of Technology, Pasadena,
California, USA}
\affiliation{Yale Quantum Institute, Departments of Applied Physics and Physics,
Yale University, New Haven, Connecticut, USA}
\orcid{0000-0002-0335-9508}

\maketitle
\selectlanguage{english}%
\begin{abstract}
This work derives an analytical formula for the asymptotic state---the
quantum state resulting from an infinite number of applications of
a general quantum channel on some initial state. For channels admitting
multiple fixed or rotating points, conserved quantities---the \textit{left}
fixed/rotating points of the channel---determine the dependence of
the asymptotic state on the initial state. The formula stems from
a Noether-like theorem stating that, for any channel admitting a full-rank
fixed point, conserved quantities commute with that channel’s Kraus
operators up to a phase. The formula is applied to adiabatic transport
of the fixed-point space of channels, revealing cases where the dissipative/spectral
gap can close during any segment of the adiabatic path. The formula
is also applied to calculate expectation values of noninjective matrix
product states (MPS) in the thermodynamic limit, revealing that those
expectation values can also be calculated using an MPS with reduced
bond dimension and a modified boundary.
\end{abstract}

\section{Introduction \& outline}

A quantum channel $\A$ (also called a quantum markov chain, Kraus
map, or completely-positive trace-preserving map) is the most general
map between two quantum systems. Channels enjoy a range of applications,
primarily in the quantum information community \cite{caruso2014},
but also in studies of matrix product states \cite{Fannes1992,Perez-Garcia2006},
entanglement renormalization \cite{Giovannetti2008,Pfeifer2009},
computability theory \cite{aaronson2016}, and even biological inference
processes \cite{Lee2016}. Whenever one studies such maps, a key question
to ask is:
\begin{equation}
\begin{array}{c}
\text{\textit{What} \textit{information} \textit{from} \textit{an} \textit{initial} \textit{state }\ensuremath{\r}}\\
\textit{survives }\textit{under}\textit{ repeated}\textit{ application}\textit{ of }\A?
\end{array}\label{eq:q1}
\end{equation}
Answering this question helps determine how to properly initialize
a quantum computer \cite{pub011,Dengis2014}, optimize the preparation
of certain exotic states \cite{Lloyd2001,Kraus2008} (part of the
active field of reservoir engineering \cite{Poyatos1996}), and even
determine properties of matrix product states in the thermodynamic
limit \cite{Perez-Garcia2006}. Quantum channels have already been
engineered in, e.g., trapped ion \cite{Schindler2013} and IBM's circuit
QED \cite{Wei2018} platforms, and reliable engineering of arbitrary
channels is actively being investigated \cite{Muller2012,Shen2016,Ticozzi2017}.
Thus, it is important for the quantum community to be on firm ground
with respect to question (\ref{eq:q1}).

There have been several studies tackling question (\ref{eq:q1}),
but many of them focused on either (A) Lindbladian channels, (B) on
channels containing only fixed points (eigenmatrices with eigenvalue
$+1$) and no rotating points (eigenmatrices with eigenvalue having
modulus $+1$), and/or (C) channels admitting a full-rank steady state
(i.e., channels that do not cause an entire subspace to decay to zero
upon infinite applications). An answer to question (\ref{eq:q1})
for the most general case is different than for any one of the specific
cases (A-C). For example, due to the ability of general channels to
permute states upon only one application, an answer \cite{pub011}
for channels generated by a Lindbladian \cite{Belavin1969,Lindblad1976,Gorini1976a,Banks1984}
does not cover the general case. To my knowledge, the following closely
related works considered channels with rotating points in detail.
Reference \cite{Wolf2010b} (supplemented with a guided tour by Wolf
\cite{wolf2010}) and an appendix of Ref.~\cite{robin} derive a
limit of infinite applications of a general channel, carefully canceling
perpetually oscillating phases. There have been two different decompositions
associated with general channels: \cite[Thm. 8]{wolf2010} being
a fine-grained block-decomposition of the \textit{asymptotic subspace}---the
subspace that survives under repeated application of the channel---and
\cite[Thm. 2]{Guan2018} being a coarser algebraic decomposition of
the Kraus operators into blocks of noiseless subsystems (see Sec.~\ref{subsec:Irreducible-channels}
for details). The works \cite{Novotny2012,Burgarth2013a,Novotny2017}
as well as older literature (e.g., \cite{groh,Fannes1992}; see also
\cite{wolf2010}) consider classes of channels with rotating points,
but such channels do not admit a decaying subspace. In fact, none
of the works listed fully address conserved quantities of channels
admitting decay.

In the MPS context \cite{Perez-Garcia2006}, a state's transfer channel
$\A$ (also called a double tensor \cite{Zeng2015}) is a channel
whose Kraus operators are used to construct the MPS. The limit of
infinite applications of $\A$ is intimately related to the thermodynamic
limit of the MPS. Question (\ref{eq:q1}) is then equivalent to
\begin{equation}
\begin{array}{c}
\text{\textit{What} \textit{part} \textit{of} \textit{the} \textit{MPS and its boundary conditions}}\\
\textit{survive }\textit{in the thermodynamic limit}?
\end{array}\label{eq:q2}
\end{equation}
The exact connection between quantum channels and MPSs has been well-studied
for the case when the MPS is injective---when its transfer channel
has only one fixed point and no rotating points. However, non-injective
MPS (with $\A$ having multiple fixed/rotating points) arise naturally
in physical systems, ranging from anti-ferromagnets and Majumdar-Ghosh
states \cite{Perez-Garcia2006} to more general spin chains \cite{Fannes1992,Bachmann2012,Bachmann2014}
and systems with topological degeneracy \cite{Perez-Garcia2008,Schuch2010}.
Non-injectivity is also ever-present in fermionic MPS \cite{Bultinck2017}
due to fermion parity preservation. In such MPS, effects due to the
interaction of ``twisted'' boundary conditions with decaying bond
degrees of freedom can persist in the thermodynamic limit, thus needing
to be taken into account. Applying results about $\A$, this work
derives a consistent thermodynamic limit, improving on Ref.~\cite{Haegeman2014},
Eqs. (130-133) (cf. \cite{Cirac,Vanderstraeten2019}). I also provide
general results about expectation values, answering question (\ref{eq:q2})
and showing how to absorb any ``twisted'' boundary effects into
the boundary matrix of an MPS with smaller bond dimension.

The first part of this manuscript answers question (\ref{eq:q1})
in a manner meant to be accessible to a large part of the quantum
community. Section \ref{sec:Asymptotics-of-channels} introduces some
notation and the general asymptotic framework. In Sec. \ref{sec:Structure-of-conserved},
I start with channels $\E$ that do not admit decay, i.e., have a
full-rank fixed-point. In Sec. \ref{sec:Extending-to-general}, I
pad the Kraus operators of $\E$ to obtain the channel $\A$ admitting
a decaying subspace. It is shown how to construct the conserved quantities
of $\A$ from those of $\E$. In Sec. \ref{sec:Application:-information-preserv},
I work in known results about the finer algebraic block decomposition
of the asymptotic subspace, addressing question (\ref{eq:q1}) for
various subspaces. Special care is taken regarding asymptotics of
irreducible channels---channels with a unique fixed point and one
or more rotating points---as those differ substantially from any
other case.

The second part of this work puts the results from the first part
to use in two applications. Given a general channel $\A$, it is shown
that its restriction $\E$ to the maximal invariant subspace is sufficient
for both applications. Section \ref{sec:Adiabaticity-of-Kraus} develops
an adiabatic limit for fixed-point spaces of continuous families of
channels $\A^{(s)}$, showing that this limit quickly reduces to the
limit of the corresponding Lindbladian channels $\E^{(s)}-\id$. This
answers the third question posed in this work,
\begin{equation}
\begin{array}{c}
\text{\text{\textit{What are any differences between the fixed-point}}}\\
\text{\textit{adiabatic limits of Hamiltonians and channels?}}
\end{array}\label{eq:q1-1}
\end{equation}
It is shown that a gap condition need only hold for a part of the
channel, and the part acting on the decaying subspace can have its
gap close without affecting the limit. Section \ref{sec:Application:Matrix-product-state}
applies to matrix product states, where it is shown how one can calculate
expectation values in the thermodynamic limit of MPS associated with
$\A$ using MPS associated with $\E$.

\section{Asymptotics of channels\label{sec:Asymptotics-of-channels}}

The canonical form of a quantum channel $\A$ and its adjoint $\A^{\dgt}$
(a generalization of the Heisenberg picture defined under the Frobenius
norm) is \cite{Sudarshan1961,Kraus1971,Choi1975}
\begin{equation}
\A\left(\r\right)=\sum_{\ell}A^{\ell}\r A^{\ell\dg}\,\,\,\,\,\,\,\,\,\,\text{and}\,\,\,\,\,\,\,\,\,\,\A^{\dgt}\left(O\right)=\sum_{\ell}A^{\ell\dg}OA^{\ell}\,,\label{eq:kraus}
\end{equation}
where $\A$ acts on states $\r$ and $\A^{\dgt}$ on operators $O$.
The matrices $A^{\ell}$ are called the Kraus operators of $\A\equiv\left\{ A^{\ell}\right\} $,
eq. (\ref{eq:kraus}) is the Kraus form of $\A$, and the only requirement
for the channel to be trace preserving is (for $I$ identity) 
\begin{equation}
\A^{\dgt}\left(I\right)=\sum_{\ell}A^{\ell\dg}A^{\ell}=I\,.\label{eq:cptp}
\end{equation}
The Kraus operators are assumed to be $D$-dimensional here, but the
intuition presented here can be extended to certain infinite-dimensional
cases \cite{Carbone2019}. Such channels can be represented as matrices
acting on a vectorized density matrix, i.e., the $D\times D$ matrix
$\r$ written as a $D^{2}$-dimensional vector. Vectorization essentially
``flips'' the bra part in each of the outer products making up $\r$
and $\A$ is written as a $D^{2}\times D^{2}$ matrix of the form
$\hat{\A}=\sum_{\ell}A^{\ell}\ot\overline{A^{\ell}}$ acting on the
vectorized $\r$ strictly from the left (with $\overline{A}$ being
the complex conjugated $A$). This \textit{matrix or Liouville representation}
of $\A$ \cite{Caves1999} is equivalent to the Kraus representation
(\ref{eq:kraus}), and I slightly abuse notation by ignoring hats
and not distinguishing the two. 

In the matrix representation, channels can be studied in terms of
their eigenvalues and eigenmatrices. The eigenvalues of all channels
are contained in the unit disk, and this work focuses on the eigenvalues/matrices
$\st$ on the periphery of that disk, i.e.,
\begin{equation}
\A\left(\st\right)=e^{i\la}\st\,\,\,\,\,\,\,\,\,\,\,\,\,\,\text{for some real }\la\,.
\end{equation}
Such eigenmatrices are called the channel's (right) \textit{rotating
points}, and those with $\la=0$ are called \textit{fixed points}.
The rotating points do not have to be physical states themselves;
e.g., $\St_{kk^{\pr}}=|k\ket\bra k^{\pr}|$ is a rotating point of
the channel $\A(\cdot)=U(\cdot)U^{\dg}$, where $|k\ket,|k^{\pr}\ket$
are eigenstates of $U$. Since $\A$ may not be normal ($[\A,\A^{\dgt}]\neq0$),
the eigenmatrices $J$ of its adjoint --- left rotating points ---
may be different from $\St$:
\begin{equation}
\A^{\dgt}\left(J\right)=e^{-i\la}J\,.
\end{equation}
Left rotating points will be called \textit{conserved quantities}
because their expectation value is either constant or oscillates with
successive powers of $\A$, but does not decay:
\begin{equation}
\tr\{J\A^{n}(\r)\}=\tr\{\A^{\dgt n}(J)\r\}=e^{-in\la}\tr\{J\r\}\,.
\end{equation}
The general block structure of $\varPsi$'s is already well-known
\cite{robin,Lindblad1999,BlumeKohout2008,baumr,Carbone2015}, and
here the focus is on the structure of the $J$'s. It is important
to note that there are as many conserved quantities as there are rotating
points (more technically, the Jordan normal form of $\A$ contains
only trivial Jordan blocks for all eigenvalues on the periphery of
the unit disk; see, e.g., Prop 6.2 in Ref. \cite{wolf2010}). All
channels have at least one fixed point \cite{Evans1978} with corresponding
conserved quantity being the identity, always conserved due to eq.
(\ref{eq:cptp}).

In the limit of many applications of $\A$, all eigenmatrices with
eigenvalues not on the periphery of the unit disk will become irrelevant
and all that will be left of the channel is the projection onto the
subspace spanned by the rotating points. The collective effect of
many applications of $\A$ is quantified by the channel's \textit{asymptotic
projection} $\ppp$,
\begin{equation}
\ppp(\cdot)\equiv\lim_{m\rightarrow\infty}\A^{\a_{m}}(\cdot)\,,\label{eq:asproj}
\end{equation}
which projects onto the eigenspace of the peripheral spectrum of the
channel. The extra parameter $\a$ allows one to take the limit in
such a way as to remove the eigenvalues $e^{i\la}$ arising from application
of $\A$ on $\r$. For any $\la=\frac{2\pi}{N}n$ (for some positive
integers $n,N$), rotating points of $\A$ are fixed points of $\A^{N}$,
so one simply takes $\a_{m}=Nm$ to get rid of the extra phases. In
this context, one can think of $\A$ as being an ``asymptotic''
root of $\ppp$ \cite{Bhat2018}. Other $\la$ which are not rational
multiples of $2\pi$ can similarly be removed to arbitrary accuracy
\cite{robin,Wolf2010b,wolf2010} by remembering that irrational numbers
are limits of sequences of rationals. The above limit is a direct
generalization of the large time limit of Lindbladian channels $\A_{t}=e^{t\L}$
for some Lindbladian $\L$. However, in that case, $\lim_{t\rightarrow\infty}e^{t\L}$
produces residual unitary evolution which is not so easily removed
by clever manipulation of the limit.

The asymptotic projection is expressible in terms of (superoperator)
projections onto the eigenspaces of the rotating points,
\begin{equation}
\ppp(\r)=\sum_{\la,\m}\st_{\la\m}\tr\left\{ J^{\la\m\dg}(\r)\right\} \,,\label{eq:ap}
\end{equation}
where the rotating points are indexed by their eigenvalue $e^{i\la}$
and $\m$ counts any degeneracies for each $\la$. In that sense,
conserved quantities are as important as fixed points despite being
less well-understood. Conveniently, the rotating points and their
corresponding conserved quantities can be made biorthogonal, $\tr\{J^{\la\m\dg}\st_{\varTheta\n}\}=\d_{\la\varTheta}\d_{\m\n}$.
The $\st$'s thus determine the basis elements of a generalized Bloch
vector \cite{alicki_book,schirmer} of the asymptotic state $\ppp(\r)$
while the $J$'s determine the coefficients of said Bloch vector.
The biorthogonality condition easily implies that $\ppp$ is really
a projection --- $\ppp^{2}=\ppp$.

The asymptotic projection for a channel with unique fixed point acts
as $\ppp(\r)=\st\tr\{\r\}=\st$. Channels with more non-trivial $\ppp$
are therefore those with multiple fixed or rotating points. As a simple
example of such a channel, consider $\A=\{A\}$ acting on $2\times2$
matrices with one Kraus operator $A=\text{diag}\{1,e^{i\t}\}$. Such
a channel sports two fixed points, the identity and the Pauli matrix
$Z$, and two rotating points $\s_{\pm}$ with eigenvalues $\la=\pm\t$.
In fact, since there is only one Kraus operator, such a channel is
actually unitary. For a non-unitary example, set $\t=\pi$ (so $A=Z$)
and add the Pauli matrix $X$ as another Kraus operator {[}normalizing
both $A$'s by $\frac{1}{\sqrt{2}}$ to satisfy trace preservation
(\ref{eq:cptp}){]}. This channel has the identity as the unique fixed
point and $Y$ as the only rotating point with $\la=\pi$. Since both
Kraus operators are Hermitian, the left and right fixed points are
the same; we will see examples when they are not later. Other examples
of $\ppp$ come from recovery maps in quantum error-correction, which
take a state which has undergone an error and project it back into
the protected subspace of the quantum code \cite{Ippoliti2014}.

\section{Faithful channels\label{sec:Structure-of-conserved}}

This part focuses on channels that do not contain a decaying subspace.
This means that no populations $|\psi\ket\bra\psi|$ decay completely
to zero under many applications of the channel: $\bra\psi|\R_{\E}(|\psi\ket\bra\psi|)|\psi\ket\neq0$
for all states $|\psi\ket$, a channel $\E$, and its asymptotic projection
$\R_{\E}$. Equivalently, the channel has to have one fixed point
$\r$ which is of full rank ($\bra\psi|\r|\psi\ket>0$ for all $|\psi\ket$).
The structural differences between such channels and channels which
do admit decay warrant a special definition:
\begin{defn*}
A channel $\E\equiv\{E_{\ell}\}$ is \textit{faithful }if it admits
a full-rank (i.e., faithful) fixed point $\r$. In other words,
\begin{equation}
\exists\,\r>0\,\,\text{ such that }\,\,\E(\r)=\r\,.
\end{equation}
\end{defn*}
Here, I always use $\E$ to denote faithful channels and later show
how $\E$ can be extended to channels $\A$ which act on a larger
Hilbert space and admit a decaying subspace. In this sense, $\E$
is the faithful channel of $\A$. Note that the number of fixed points
is independent of this condition, and Table \ref{tab:Some-types-of-channels}
relates this definition to others.

\begin{table}
{\scriptsize

\begin{tabular}{>{\centering}p{2.3cm}>{\centering}p{1.5cm}>{\centering}p{1.5cm}>{\centering}p{1.3cm}}
\toprule 
 & Fixed point unique? & $\exists$ full-rank fixed point? & $\exists$ rotating point?\tabularnewline
\midrule
\midrule 
ergodic \cite{Oseledets1984,Raginsky2002,Raginsky2002a,Burgarth2007} & Yes &  & \tabularnewline
faithful {[}here{]} &  & Yes & \tabularnewline
irreducible \cite{Davies1970,wolf2010} & Yes & Yes & \tabularnewline
mixing \cite{Burgarth2007} & Yes &  & No\tabularnewline
primitive \cite{Sanz2010,Wolf2010b} & Yes & Yes & No\tabularnewline
\bottomrule
\end{tabular}

}

\caption{\label{tab:Some-types-of-channels}Various types of channels. A blank
entry means there is no requirement for that definition. For Lindbladians,
mixing is also known as relaxing \cite{Burgarth2013a} and faithful
is also known as minimal \cite{ABFJ}. Primitive is equivalent to
strongly irreducible \cite{Sanz2010} and irreducible and aperiodic
\cite{Guan2018a}.}
\end{table}

The first statement is regarding the relationship between the conserved
quantities $J$ and the Kraus of operators of $\E$. It is a generalization
of a theorem for fixed points of faithful channels \cite{robin,Kribs2003,Choi2006,Gheondea2016,Gheondea2018},
which states that a conserved quantity $J$ with eigenvalue $\la=0$
commutes with all of the Kraus operators. It is shown that conserved
quantities with $\la\neq0$ commute up to a phase. For the aforementioned
example $\E=\{E\}$ with $E=\text{diag}\{1,e^{i\t}\}$, the conserved
quantity $\s_{+}$ satisfies $\s_{+}E=e^{-i\t}E\s_{+}$. This turns
out to be true for all faithful channels. We state the result below,
relegating all proofs for the appendix.

\begin{numtheorem}{\hyperref[prop:1]{Proposition 1.}}

Let $\E=\left\{ E_{\ell}\right\} $ be a faithful channel. Let $J$
be a conserved quantity of $\E$, i.e., 
\begin{equation}
\E^{\dgt}\left(J\right)=e^{-i\la}J
\end{equation}
 for some real $\la$. Then, for all $\ell$,
\begin{equation}
JE_{\ell}=e^{-i\la}E_{\ell}J\,.\label{eq:com}
\end{equation}

\end{numtheorem}

Next, I show that this formula imposes some restrictions on the peripheral
spectrum of $\E$, and that unitary $J$ can be thought of as symmetries
in a limited extension of Noether's theorem.

\subsection{Spectral restrictions\label{subsec:Spectral-restrictions}}

Assuming $\E^{\dgt}(J_{1})=e^{-i\la{}_{1}}J_{1}$ and $\E^{\dgt}(J_{2})=e^{-i\la{}_{2}}J_{2}$,
Eq.~(\ref{eq:com}) implies that $\E^{\dgt}(J_{1}J_{2})=e^{-i(\la{}_{1}+\la_{2})}J_{1}J_{2}$.
Combined with the fact that there must be $\leq D^{2}$ conserved
quantities, there are sometimes constraints on $\la$ such that there
remain finitely many eigenvalues. Thus, I suggest that we divide conserved
quantities into two types: nilpotent ones $\jn$ and diagonalizable
ones $\jd$. Since $\jn^{N}=0$ for $N\leq D$, the dimension of the
Hilbert space, there are no restrictions on their eigenvalue. For
a unitary channel $\E=\{U\}$, such conserved quantities are the coherences
$|k\ket\bra k^{\pr}|$ between eigenstates $|k\ket\neq|k^{\pr}\ket$
of $U$. On the other hand, projections $|k\ket\bra k|$ onto eigenstates
of $U$ are examples of $\jd$ with $\la=0$. While this is the only
possible value for $\jd$ of unitary channels, general channels admit
special $\jd$ with $\la\neq0$. For example, the channel $\E=\{X/\sqrt{2},Z/\sqrt{2}\}$
admits conserved quantities $\jd\in\{I,Y\}$ with eigenvalues $\pm1$,
respectively. More generally, it turns out these special $\jd$ can
have only $N$th root-of-unity eigenvalues, with $N$ tightly bounded
by the dimension of the range of $\jd$.

\begin{numtheorem}{\hyperref[prop:2]{Proposition 2.}}

Let $\E=\left\{ E_{\ell}\right\} $ be a faithful channel. Let $\jd$
be such that $\E^{\dgt}(\jd)=e^{-i\la}\jd$ for some real $\la$ and
assume $\jd$ is diagonalizable. Then, there exists an integer $n$
such that 
\begin{equation}
\la=\frac{2\pi}{N}n\,\,\,\,\,\,\,\text{for some}\,\,\,\,\,\,\,N\leq\left\Vert \jd\right\Vert _{1}\,,
\end{equation}
where $\left\Vert \,\,\,\right\Vert _{1}$ is the trace norm.

\end{numtheorem}

Splitting the set of $\jd$ into the ordinary $\la=0$ case and the
above special case, one can classify conserved quantities and their
corresponding rotating points (it is implied that $\la\neq0$ in the
latter two types):

\begin{enumerate}[I]

\item Ordinary conserved quantities: $\la=0$

\item Nilpotent conserved quantities: $\la\in\mathbb{R}$

\item Diagonalizable conserved quantities: $\la=\frac{2\pi}{N}n$.

\end{enumerate}

Types I and II exist for Lindbladian evolutions, but type III is unique
to general channels. Type III quantities affect the block-decomposition
of rotating points into noiseless factors; this is addressed in Sec.~\ref{sec:Application:-information-preserv}.

\subsection{Limited Noether-type theorem}

Let us assume a unitary conserved quantity, $J^{\dg}J=JJ^{\dg}=I$,
and show that the above two propositions extend known results (\cite{wolf2010},
Prop. 6.7) from irreducible to faithful channels. Proposition \ref{prop:1}
readily implies that $\E$ is covariant (more specifically, invariant
or symmetric) under $J$,
\begin{equation}
J\E(\r)J^{\dg}=\E(J\r J^{\dg})\,\,\,\,\,\,\forall\r\,,
\end{equation}
so conserved quantities are symmetries of the channel. Such quantities
are also important in determining how the set of channels decomposes
into path-connected subsets \cite[Thm. 5]{Szehr2016} (see also \cite{Bachmann2012a}).
Proposition \ref{prop:2} implies that $J^{N\leq D}=I$, so the set
$\{J^{n}\}_{n=0}^{N-1}$ forms the symmetry group $\Z_{N}$. Note
that the symmetry group is never infinite for finite dimension $D$.
Generalizing this, the set of unitary conserved quantities thus forms
a finite group under which $\E$ is covariant. This is a one-way Noether-type
theorem linking conserved quantities to symmetries (see Refs. \cite{pub011,Gough2015}
or Ref. \cite{thesis}, Ch. 2.6, for the Lindbladian analogue). Note
that this generalization is in a different direction from Ref.~\cite{Marvian2014},
which introduces measures quantifying the extent to which a quantum
state breaks a symmetry. This one-way theorem cannot be extended to
a two-way theorem because symmetries of a channel are not always conserved
quantities. A simple counterexample is the channel $\E=\{X/\sqrt{2},Z/\sqrt{2}\}$,
for which the Hadamard operation $H$ taking $X\leftrightarrow Z$
is a symmetry, but is not conserved {[}$\E^{\dgt}(H)=0${]}.

\section{Extending to general channels\label{sec:Extending-to-general}}

Now let us extend faithful channels to channels which do not contain
a full-rank fixed point. While Props. \ref{prop:1}-\ref{prop:2}
break down for general channels, the extension below implies that,
for every general channel, there is a corresponding faithful channel
for which they hold.

Any faithful channel $\E=\left\{ E_{\ell}\right\} $ can be extended
to a channel $\A=\left\{ A^{\ell}\right\} $ which contains a decaying
subspace (also, transient subspace \cite{ying2013}). Specifically,
the Kraus operators of $\A$ are
\begin{equation}
A^{\ell}=\begin{pmatrix}E_{\ell}\vspace{4pt} & A_{\ur}^{\ell}\\
0 & A_{\lr}^{\ell}
\end{pmatrix}\equiv\begin{pmatrix}A_{\ul}^{\ell}\vspace{4pt} & A_{\ur}^{\ell}\\
0 & A_{\lr}^{\ell}
\end{pmatrix}\,.\label{eq:bl}
\end{equation}
The zero is necessary by assumption: for $\E$ to act on the largest
invariant subspace, $\A$ cannot take states out of that subspace.
The dimensions of the square matrices $E_{\ell}$ and $A_{\lr}^{\ell}$
can differ, and the bounds of $\ell$ can change by padding the same
$E$ with two different pairs of matrices in $\urbig$ (``upper right'')
and $\lrbig$ (``lower right'') to make two different $A$'s. The
zero matrix in $\llbig$ is necessary to make sure that $\ulbig$
is the largest invariant subspace; thus, all rotating points of $\A$
are the same as those of $\E$. In addition, $\A$ needs to be a legitimate
channel, i.e., satisfy eq. (\ref{eq:cptp}). Writing out the $A^{\ell}$'s
in blocks {[}as in eq. (\ref{eq:bl}){]} yields the conditions\begin{subequations}\label{eq:conds}
\begin{align}
\sum_{\ell}A_{\ul}^{\ell\dg}A_{\ul}^{\ell} & =\pp\label{eq:conds1}\\
\sum_{\ell}A_{\ul}^{\ell\dg}A_{\ur}^{\ell} & =0\label{eq:conds2}\\
\sum_{\ell}(A_{\ur}^{\ell})^{\dg}A_{\ur}^{\ell}+A_{\lr}^{\ell\dg}A_{\lr}^{\ell} & =\qq\,,
\end{align}
\end{subequations}where $\qq$ is the projection on $\lrbig$ and
$\pp=I-\qq$ is the projection onto $\ulbig$ (with $\tr\left\{ P\right\} \equiv D$).
For each faithful channel $\E$, there are an infinite number of possible
extensions $\A$. Conversely, an arbitrary channel $\A$ either is
a faithful channel or contains one, i.e., once the largest invariant
subspace $\ulbig$ is determined (via, say, an algorithm \cite{Cirillo2015}),
one will obtain the block decomposition (\ref{eq:bl}). The remaining
two completely positive maps associated with this decomposition of
$\A$, $\{A_{\ur}^{\ell}\}$ and $\{A_{\lr}^{\ell}\}$, are both trace-decreasing.

Now let us develop the required notation. Just like $\pp$ and $\qq$
split the Hilbert space into two parts, they can be used to split
the space of operators on a Hilbert space into four ``corners''
$\{\ulbig,\urbig,\llbig,\lrbig\}$ \cite{ABFJ}. Each of the four
corners corresponds to its own superoperator projection. For example,
\begin{equation}
\R_{\ur}(O)\equiv\pp O\qq\equiv O_{\ur}
\end{equation}
for any operator $O$. The other three projections are defined accordingly.
One can graphically determine which corner a product of operators
belongs to by multiplying their blocks as matrices (e.g., $A_{\ll}B_{\ur}\in\lrbig$).
Moreover, the four-corners projections add graphically ($\R_{\ul}+\R_{\lr}\equiv\R_{\di}$)
and are Hermitian ($\R_{\emp}^{\dgt}=\R_{\emp}$). Analogous to studying
operators in terms of their matrix elements, one can study superoperators
in terms of their four-corners decomposition. For example,
\begin{equation}
\R_{\ul}\A\R_{\lr}(\r)=\pp\A\left(\qq\r\qq\right)\pp=\sum_{\ell}A_{\ur}^{\ell}\r_{\lr}(A_{\ur}^{\ell})^{\dg}\label{eq:tr}
\end{equation}
is the map $\{A_{\ur}^{\ell}\}$ which transfers $\r_{\lr}$ from
$\lrbig$ to $\ulbig$. ``Diagonal'' elements are denoted as $\A_{\emp}\equiv\R_{\emp}\A\R_{\emp}$
for convenience, so the faithful channel $\E\equiv\R_{\ul}\A\R_{\ul}$
and similarly $\{A_{\lr}^{\ell}\}\equiv\R_{\lr}\A\R_{\lr}$.

With conditions (\ref{eq:bl}) and (\ref{eq:conds}), $\A$ contains
a decaying subspace of dimension $\tr\left\{ \qq\right\} $ and the
same rotating points as $\E$. But what about the conserved quantities?
Those are not the same because, by trace preservation, they need to
make sure that all state populations (and sometimes some coherences)
in $\lrbig$ are transferred to $\ulbig$. For example, the identity
is (always) a conserved quantity of $\A$, but the analogous conserved
quantity of $\E$ is $\pp$. Denoting the conserved quantities of
$\E$ as $J_{\ul}$, it will now be shown how to extend them to form
$J$, the conserved quantities of $\A$. Having defined this notation,
it is easy to write out the conserved quantities of the extended channel
$\A$. 

\begin{numtheorem}{\hyperref[prop:3]{Proposition 3.}}

The conserved quantities of $\A$ corresponding to eigenvalues $e^{i\la}$
are
\begin{equation}
J=J_{\ul}+J_{\lr}=J_{\ul}+\RR_{\lr}^{(\la)\dgt}\A^{\dgt}(J_{\ul})\,,\label{eq:main}
\end{equation}
where $J_{\ul}$ are the conserved quantities of $\A_{\ul}=\E$ and
\begin{equation}
\RR_{\lr}^{(\la)}=-\left(\A-e^{i\la}\right)_{\lr}^{-1}\,.\label{eq:resolvent}
\end{equation}

\end{numtheorem}

An important corollary of the above proposition is that $J_{\of}=0$.
After plugging in this formula for $J$ into $\ppp$ (\ref{eq:ap}),
this means that the asymptotic projection has only two pieces:
\begin{equation}
\ppp=\R_{\ul}\ppp\R_{\di}\equiv\ps+\ppp\R_{\lr}\,,\label{eq:ppp}
\end{equation}
where the \textit{faithful projection} (for Lindbladians, minimal
projection \cite{ABFJ})
\begin{equation}
\ps(\cdot)\equiv\ppp\R_{\ul}(\cdot)=\sum_{\la,\m}\St_{\la\m}\tr\{J_{\ul}^{\la\m\dg}\cdot\}
\end{equation}
is the asymptotic projection of the faithful channel $\E$. The piece
$\ps$ is responsible for preserving parts of an initial state $\r$
which is in $\ulbig$ while the piece $\ppp\R_{\lr}$ is a channel
mapping states from $\lrbig$ onto the subspaces spanned by the rotating
points of $\A$, all located in $\ulbig$. The key result here is
that the rotation induced by $\la$, besides inducing phases on the
rotating points, also contributes to the decay of information from
$\lrbig$ into $\ulbig$. Namely, the inverse of the resolvent (\ref{eq:resolvent})
modulates the decoherence induced during the decay in a way that depends
on how close the eigenvalues of $\A_{\lr}$ are to the phases $e^{i\la}$.
The $\la=0$ case reduces to known results (\cite{robin}, Lemma 5.8;
\cite{Cirillo2015}, Prop. 7),
\begin{equation}
\ppp\R_{\lr}=\ps\A(\id-\A)_{\lr}^{-1}\,,
\end{equation}
where the resolvent can be thought of as the quantum version of the
fundamental matrix from classical Markov chains \cite{markov_book}.
These formulas also reduce to the Lindbladian result (\cite{ABFJ},
Prop. 3) if we let $\A=e^{\L}\rightarrow\id+\L$ for some Lindbladian
$\L$ and $e^{-i\la}\rightarrow1-i\la$. In the Lindblad case, some
dependence on $\la$ can be canceled by properly tuning $\L_{\lr}$
(\cite{thesis}, Sec. 3.2.3).

Directly applying Prop.~\ref{prop:3} allows us to find the asymptotic
\cite{Cirillo2015} (also, reachability \cite{ying2013}) probabilities
of a given initial state $\r$ to reach a particular subspace of $\ulbig$.
The new result here is determination of the \textit{coherences} reached
by $\r$, assuming knowledge of the left ($J_{\ul}^{\la\m}$) and
right ($\St_{\la\m}$) rotating points of $\E$. To determine the
coefficient $c_{\la\m}$ in the asymptotic state $\ppp(\r)=\sum_{\la\m}c_{\la\m}\St_{\la\m}$
(\ref{eq:ap}), instead of applying $\A$ a sufficiently large number
of times to determine $\ppp$, simply calculate
\begin{equation}
c_{\la\m}=\tr\left\{ J_{\ul}^{\la\m\dg}\r_{\ul}\right\} +\sum_{\ell}\tr\left\{ A_{\ur}^{\ell\dg}J_{\ul}^{\la\m\dg}A_{\ur}^{\ell}\RR_{\lr}^{(\la)}(\r_{\lr})\right\} ,
\end{equation}
where I used eq.~(\ref{eq:tr}) and cyclic invariance.

\section{Adding fixed-point structure\label{sec:Application:-information-preserv}}

We have yet to consider the fine-grained structure of the asymptotic
subspace. This section builds up the most general asymptotic subspace
from its minimal ingredients: channels admitting decoherence-free
subspaces and irreducible channels. The latter can have Type III (see
Sec.~\ref{subsec:Spectral-restrictions}) rotating points in the
general case, making the usual noiseless subsystem decomposition significantly
different from the Lindbladian or unitary cases. This section concludes
with an algorithm that outputs a properly organized $\ppp$ given
a channel $\A$.

\subsection{Decoherence-free subspaces\label{subsec:Decoherence-free-subspaces}}

Let us assume that now all of $\ulbig$ consists of rotating or fixed
points, so $\A_{\ul}=\E$ is a unitary channel. An example of this
case is $\A_{\ul}=\left\{ E\right\} $, where $A_{\ul}=E=\text{diag}\{1,e^{i\t}\}$
is the Kraus operator that mentioned before. The necessary and sufficient
condition on the $A$'s (\ref{eq:bl}) for this to hold is
\begin{equation}
A_{\ul}^{\ell}=a_{\ell}U\label{eq:dfs}
\end{equation}
for some unitary $U$, real $a_{\ell}$, and such that $\sum_{\ell}|a_{\ell}|^{2}=1$
to satisfy the condition (\ref{eq:conds1}). Since there is no decay
in $\ulbig$, that portion forms a \textit{decoherence-free subspace}
(DFS) \cite{Lidar1998} and $\ps=\R_{\ul}$. The form of $A_{\ul}$
also implies that $\R_{\ul}\A\R_{\of}=0$. The rotating points and
conserved quantities of $\E$---outer products of eigenstates of
$U$---are of course the same.

Adding in decay, the form (\ref{eq:bl}) of $A$ with the above restriction
on $A_{\ul}$ generalizes the previous DFS condition from eq. (11)
of Ref. \cite{lidar2003} (see also Refs.~\cite{Karasik2008,Kamizawa2018}
for different formulations). The difference is that now $A_{\ur}$
does not have to be zero, so information from $\lrbig$ flows into
the DFS $\ulbig$. For example, in quantum error-correction, $\ulbig$
is the logical subspace, $\lrbig$ is the orthogonal error subspace,
and the piece $\ppp\R_{\lr}$ plays the role of a ``recovery channel''
which attempts to recover the leaked information after an error \cite{Ippoliti2014}.
It turns out one can remove the inverse term from $\ppp\R_{\lr}$,
putting the piece in Kraus form. Setting $A_{\lr}=0$ and $A_{\ul}=\pp$
(unitary evolution within DFS is trivial) eliminates $\A_{\lr}$ and
reduces $\ppp\R_{\lr}$ to the transfer map (\ref{eq:tr}), 
\begin{equation}
\ppp\R_{\lr}=\R_{\ul}\A\R_{\lr}\,,\label{eq:arb}
\end{equation}
with Kraus operators $A_{\ur}$. Condition (\ref{eq:conds2}) on $A_{\ur}$
reduces to $\sum_{\ell}A_{\ur}^{\ell}=0$, which is automatically
satisfied by the set of operators $\{\pm A_{\ur}^{\ell}/\sqrt{2}\}$.
However, the channel created by those operators is the same as $\{A_{\ur}^{\ell}\}$,
so $\ppp$ embeds an arbitrary recovery channel from the error subspace
$\lrbig$ to code subspace $\ulbig$.

\subsection{Irreducible channels\label{subsec:Irreducible-channels}}

On the opposite side from DFS channels are the \textit{irreducible}
channels---channels admitting a unique full-rank fixed point. We
find $\ppp$ for such channels below, including the presence of any
additional rotating points---a case unique to non-Lindbladian channels.

For Lindbladian channels, there is only one asymptotic state (\ref{eq:ap})
$\ppp(\r)=\varrho$, the unique full-rank (on $\ulbig$) fixed point
of the channel. For general channels, this situation is complicated
by any Type III rotating points. An example is the channel $\E=\A_{\ul}=\{X/\sqrt{2},Y/\sqrt{2}\}$,
admitting two conserved quantities $\{I,Z\}$. Due to Prop.~\ref{prop:2},
all conserved quantities $J_{\ul}$ for irreducible channels are unitary
(on $\ulbig$) and their corresponding rotating points are simply
$\St=J_{\ul}\varrho$. Moreover, conserved quantities are powers of
a single unitary 
\begin{equation}
J_{\ul}=\sum_{n\in\Z_{N}}\o^{n}\varPi_{n}\,,\label{eq:irreducible-J}
\end{equation}
where $\o=e^{i\frac{2\pi}{N}}$ is this quantity's eigenvalue (for
some $N\leq D$, the total dimension) and $\varPi_{n}$ is the projection
onto the eigenspace of eigenvalue $\o^{n}$ {[}\citealp{wolf2010},
Thm.~6.6; see also proof of Prop.~\ref{prop:2}{]}. We can thus
represent them as $J_{\ul}^{\m}$, where $\m\in\Z_{N}$ is literally
the $\m$th power of $J_{\ul}$. This $J_{\ul}$ must commute with
$\varrho$; otherwise, $J_{\ul}\varrho J_{\ul}^{\dg}$ would be a
different fixed point, violating irreducibility of $\E$. 

Biorthogonality between the rotating points and conserved quantities
with different eigenvalues and $\tr\{\varrho\}=1$ imply that the
following must be true:
\begin{equation}
\tr\{J_{\ul}^{\m\dg}\st_{\n}\}=\sum_{n\in\Z_{N}}\o^{\left(\n-\m\right)n}\tr\left\{ \varPi_{n}\varrho\right\} =\d_{\m\n}\,.\label{eq:ortho}
\end{equation}
Satisfaction of these $N$ equations forces $\tr\{\varPi_{n}\varrho\}=\nicefrac{1}{N}$.
A striking fact is that the $\st$'s themselves are no longer orthonormal
under the Hilbert-Schmidt inner product, $\tr\{\st_{\m}^{\dg}\st_{\n}\}\not\propto\d_{\m\n}$.
This was the case for Lindbladian channels \cite{ABFJ}, and a simple
counterexample\footnote{Let $\E=\{|2\ket\bra0|,|2\ket\bra1|,{\textstyle \sqrt{\frac{2}{3}}}|1\ket\bra2|,{\textstyle \sqrt{\frac{1}{3}}}|0\ket\bra2|\}$
act on $\{|0\ket,|1\ket,|2\ket\}$. This channel has a fixed point
$\varrho=\text{diag}\{1/6,1/3,1/2\}$, and $J=\text{diag}\{1,1,-1\}$.
Eigenmatrices $\varrho$ and $J\varrho$ are not orthogonal.} shows this is not true generally. 

There is an alternative characterization \cite{groh,Wolf2010b} of
the subspace spanned by the rotating points, popularized by Wolf \cite[Thm. 6.16]{wolf2010}.
Instead of using $\st_{\m}$, one uses the basis 
\begin{equation}
\stt_{n}\equiv\varPi_{n}\varrho=\varPi_{n}\varrho\varPi_{n}
\end{equation}
of $\varrho$ projected onto eigenspaces of $J$. For our example
$\E$, this basis is $\stt_{0}=|0\ket\bra0|$ and $\stt_{1}=|1\ket\bra1|$,
and $\E$ permutes one with the other. This characterization allows
one to decompose the space of rotating points into finer blocks (in
the example, two 1-by-1 blocks). The $\stt_{n}$ are not eigenmatrices---they
are permuted by $\E$ (see proof of Prop.~\ref{prop:2})---but they
are fixed points of $\E^{N}$. However, unlike $\st_{\m}$, this basis
is orthogonal.

As for the asymptotic state (\ref{eq:ap}), despite there being only
one \textit{fixed} point, the rotating points also have to be included:
\begin{align}
\ps(\cdot) & =\sum_{\m\in\Z_{N}}\tr\{J_{\ul}^{\m\dg}(\cdot)\}J_{\ul}^{\m}\varrho=N\sum_{n\in\Z_{N}}\tr\{\varPi_{n}(\cdot)\}\varPi_{n}\varrho\,,\label{eq:asymp-irreducible}
\end{align}
where I plugged in eq.~(\ref{eq:irreducible-J}) and simplified.
Extending from $\ps$ to the full $\ppp$ just substitutes $J_{\di}^{\m}$
for the first $J_{\ul}^{\m}$, with $J_{\lr}^{\m}$ determined by
Prop.~\ref{prop:3}.

\subsection{Noiseless subsystems}

The DFS case in Sec.~\ref{subsec:Decoherence-free-subspaces} and
the irreducible case in Sec.~\ref{subsec:Irreducible-channels} are
the two building blocks out of which one constructs the most general
rotating point structures. Focusing on the $\ulbig$ space, a tensor
product of a unitary channel and an irreducible channel $\{R_{\ell}\}$
has Kraus operators (\ref{eq:bl}) \cite{Guan2018}
\begin{equation}
A_{\ul}^{\ell}=U\otimes R_{\ell}\,.
\end{equation}
In Lindbladians, such a combination of a DFS and an auxiliary subspace
with a unique fixed point is collectively called a \textit{noiseless
subsystem (NS)} \cite{Knill2000}. Its asymptotic state would then
be expressible in the matrix basis $|k\ket\bra k^{\pr}|\ot\varrho$,
where $|k\ket,|k^{\pr}\ket$ are eigenstates of $U$ with eigenvalues
$\l_{k},\l_{k^{\pr}}$ and $\varrho$ is the unique fixed point of
$\{R_{\ell}\}$ (e.g.,~\cite{pub011}). However, the Type III rotating
points of $\{R_{\ell}\}$ necessitate another index in the general
case. Now, the rotating points are $\St_{kk^{\pr}\m}=|k\ket\bra k^{\pr}|\ot J^{\m}\varrho$,
satisfying the eigenvalue equation
\begin{equation}
\A(\St_{kk^{\pr}\m})=e^{i(\l_{k}-\l_{k^{\pr}}+\frac{2\pi}{N}\m)}\St_{kk^{\pr}\m}\,.
\end{equation}
This makes contact with {[}\citealp{Wolf2010b}, Thm.~9{]}, showing
that the most general peripheral eigenvalues are combinations of eigenvalues
of a unitary $U$ and a root of unity stemming from a Type III rotating
point. It is not generally possible to make an orthonormal basis of
rotating points when the irreducible factor admits rotating points
(see previous subsection).

Noiseless blocks can further be stacked to form the most general asymptotic
subspace, corresponding to Kraus operators
\begin{equation}
A_{\ul}^{\ell}=\bigoplus_{\varkappa}U_{\varkappa}\otimes R_{\ell,\varkappa}\,,\label{eq:decomp}
\end{equation}
where $U_{\varkappa}$ is unitary and $\{R_{\ell,\varkappa}\}$ for
each $\varkappa$ is irreducible \cite{Guan2018}. This blocks-of-factors
structure or \textit{shape} of $A_{\ul}^{\ell}$ is the most general
form of an information-preserving structure \cite{robin}. Rotating
points and conserved quantities of the form from the previous subsection
can be constructed for each block to form a \textit{canonical basis}
(i.e., a basis respecting the block structure). It is well-known among
experts (see, e.g., \cite{Rahaman2017}) that the conserved quantities
$\{J_{\ul}\}$ form a \textit{matrix algebra} --- a vector space
(where the vectors are matrices) that is closed under multiplication
and the conjugate transpose operation. It is important to keep in
mind that all of this extra structure in $\ulbig$ does not put any
constraints on the remaining parts $\{A_{\ur},A_{\lr}\}$ of $\A$,
the extension of $\E$ (as long as eqs.~(\ref{eq:conds}) are satisfied).
The extended conserved quantities $J$ do \textit{not} have to form
a matrix algebra.

\subsection{How to find $\protect\ppp$\label{subsec:Algorithm-for-finding}}

There exist several algorithms to determine the shape (\ref{eq:decomp})
of $\A$ \cite{robin,Holbrook2003,Choi2006,Knill2006,Maehara2010,Wang2013,Guan2018}.
A straightforward way \cite{robin} to find the form (\ref{eq:decomp})
for a general channel $\A$ is to diagonalize $\A$ and apply standard
matrix algebra techniques \cite{Holbrook2003,Maehara2010} to find
a canonical basis for the algebra of conserved quantities in $\ulbig$.
Using Prop. \ref{prop:3}, I slightly extend the algorithm from Ref.
\cite{robin} to one that finds and organizes not just the conserved
quantities restricted to $\ulbig$, but the full conserved quantities
as well. Once again, the main new inclusion is the determination of
conserved quantities whose eigenvalue is modulus one (as opposed to
exactly one).
\begin{lyxalgorithm*}
Finding and organizing $\ppp$

Find the rotating points $\St$ and conserved quantities $J$ by diagonalizing
$\A$

Construct $\ppp$ and $\pp$, the projection onto $\textnormal{range}\{\ppp(I)\}$

Find the projected conserved quantities $J_{\ul}\equiv\pp J\pp$

Decompose the algebra spanned by $J_{\ul}$ into canonical form using,
e.g., Refs.~\cite{Holbrook2003,Maehara2010}

Determine a canonical basis $\St$ for the rotating points and $J_{\ul}$
for the conserved quantities

Extend $J_{\ul}$ to $J$ via Prop. \ref{prop:3}.
\end{lyxalgorithm*}
Note that $\ulbig$ is the range of $\ppp(I)$, i.e., $\ppp(I)\propto\pp$,
because $I$ is dual to the maximally mixed fixed point $\frac{1}{\tr\{\pp\}}\pp$
and is the only conserved quantity with nonzero trace.

\section{Application: Adiabaticity\label{sec:Adiabaticity-of-Kraus}}

In this section, an adiabatic limit is derived for continuously parameterized
families of channels admitting only fixed points. It is shown that
the spectral gap of $\A_{\lr}$ can be closed for any segment of the
adiabatic path without affecting the adiabatic evolution of the fixed
points. The only gap that must remain open is that of $\A_{\of}$.

\subsection{From channels to Lindbladians}

Consider a continuously parameterized family of channels $\A^{(s)}$
with $s\in[0,1]$ and with no rotating points. Instead of thinking
of these as repeated applications of random channels \cite{Nechita2012,Bruneau2014},
we assume that $\A^{(s)}$ is smoothly varying. Starting with an initially
fixed-point state of $\rout=\A^{(0)}(\rout)$, we determine how this
state evolves under the map
\begin{equation}
\T(\rout)\equiv\lim_{N\rightarrow\infty}\A^{(1)}\cdots\A^{(\frac{1}{N})}\A^{(\frac{0}{N})}(\rout)\,.\label{eq:krausadia}
\end{equation}
This is a discrete channel version of the usual quantum-mechanical
adiabatic limit \cite{Born1928,Kato1950}, where the generator of
unitary time evolution is slowly varied during evolution. Those familiar
with unitary quantum mechanics would expect that, upon evolution in
a closed path $\A^{(1)}=\A^{(0)}$, one remains in the fixed-point
subspace of $\A^{(s)}$ and returns to the same state up to a geometric
phase \cite{Berry1984} or unitary matrix \cite{Wilczek1984} operation.
Such a limit has indeed been extended to fixed-point subspaces of
Lindbladians (e.g., {[}\citealp{Avron2012b},~Thm.~2.6{]}), and
the few simple manipulations below will convert the above limit of
discrete maps into this case. However, note that initializing in a
subspace other than the one of slowest decay is not as straightforward
\cite{Nenciu1992,Berry2011,Uzdin2011,Holler2019}.

To reduce the above, we use the calculus of differences \cite{concrete}.
For $n\in\{0,1,\cdots,N$\}, define the map $\r^{(\frac{n+1}{N})}=\A^{(\frac{n+1}{N})}(\r^{(\frac{n}{N})})$.
Taking the discrete derivative yields
\begin{eqnarray}
\!\!\!\!\!\!\!\!\!d\r^{(n)} & \equiv & \frac{\r^{(\frac{n+1}{N})}-\r^{(\frac{n}{N})}}{1/N}=N\left(\A^{(\frac{n+1}{N})}-\id\right)(\r^{(\frac{n}{N})})\,.
\end{eqnarray}
At this point, one can proceed to derive a discrete version of the
adiabatic expansion of Ref.~\cite{Avron2012b}. However, it is simpler
to take the continuum limit directly, yielding the evolution equation
$\frac{1}{N}\frac{d\r^{(s)}}{ds}=\L^{(s)}(\r^{(s)})$, generated by
the Lindbladian 
\begin{equation}
\L^{(s)}=\A^{(s)}-\id\,.
\end{equation}
Therefore, adiabaticity of Kraus map families $\A^{(s)}$ is solely
determined by their corresponding Lindbladians $\L^{(s)}$. Using
the aforementioned results \cite{Avron2012b}, the asymptotic solution
to eq.~(\ref{eq:krausadia}) is
\begin{equation}
\T(\rout)=\prod_{s\in[0,1]}\ppp^{(s)}(\rout)+O\left(1/N\right)\,,
\end{equation}
where $\prod$ is the time-ordered product of $\ppp^{(s)}$. The evolution
can also be calculated using the Berry connection \cite{Sarandy2006,ABFJ}
\begin{equation}
A_{\a,\m\n}=\tr\left\{ J^{\m\dg}\p_{\a}\st_{\n}\right\} \,,\label{eq:connect}
\end{equation}
where one can parameterize $\p_{s}=\sum_{\a}v_{\a}\p_{\a}$ into various
directions $\a$ in parameter space with velocities $v_{\a}$, and
$\{\st,J\}$ are the fixed points and conserved quantities, respectively.
This connection allows us to express the closed-loop adiabatic evolution
as a \textit{holonomy }\cite{Simon1983}
\begin{equation}
\T(\rout)=\mathbb{P}e^{-\sum_{\a}\int d\a A_{\a}}(\rout)+O\left(1/N\right)\,,\label{eq:adiabatic-evolution-A}
\end{equation}
telling us how a state in the instantaneous fixed-point subspace has
rotated as the entire subspace is parallel transported in the parameter
space (with $\mathbb{P}$ denoting a path-ordered integral).

Now let us add in the $\empbig$ structure of $\A$, meaning that
for each $s$ there is an additional instantaneous set of projections
$\R_{\emp}^{(s)}$. Assuming that all fixed points lie in $\ulbig$,
it was shown \cite{ABFJ} that the $O(1/N)$ correction term depends
on the spectrum of $\L_{\of}$ and not of $\L_{\lr}$. Below, we go
further and show that the \textit{dissipative gap} of $\L_{\lr}$---the
nonzero eigenvalue with the smallest real part---can even be closed
without affecting the adiabatic evolution of the fixed points in $\ulbig$.

\begin{figure}
\centering{}\includegraphics[width=0.6\columnwidth]{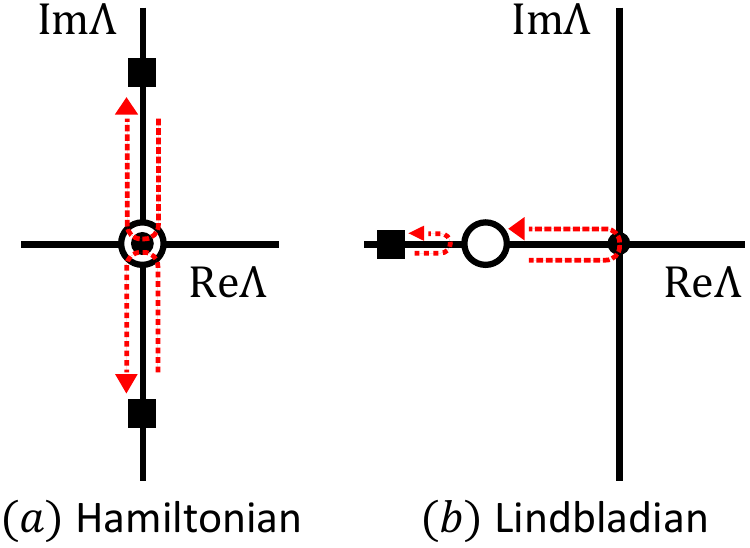}\caption{\label{fig:path}\textbf{(a)} The path $s\in[0,1]$ traced by the
four eigenvalues $\Lambda$ of the superoperator corresponding to
the Hamiltonian $H=\protect\half\cos^{2}\pi s(|0\protect\ket\protect\bra0|-|1\protect\ket\protect\bra1|)$
with eigenstates $|0\protect\ket,|1\protect\ket$. The two populations
$|0\protect\ket\protect\bra0|$ (labeled by $\bullet$) and $|1\protect\ket\protect\bra1|$
($\bigcirc$) are fixed points of the evolution while the two coherences
$|0\protect\ket\protect\bra1|,|1\protect\ket\protect\bra0|$ ($\blacksquare$)
are eigenmatrices with the energy difference as the eigenvalue. As
$|0\protect\ket$ and $|1\protect\ket$ become degenerate at $s=\protect\half$,
the coherence eigenvalues go to zero, and so transitions between $|0\protect\ket$
and $|1\protect\ket$ can occur. \textbf{(b)} A path traced by the
four eigenvalues of the Lindbladian with jump operator (\ref{eq:jump}).
The coherences remain separated from the origin when both $|0\protect\ket\protect\bra0|$
and $|1\protect\ket\protect\bra1|$ are fixed points at $s=\protect\half$,
so transitions between $|0\protect\ket$ and $|1\protect\ket$ are
suppressed.}
\end{figure}

\subsection{The gap can close: a simple example}

Let us first consider a simple example of a channel whose gap in $\lrbig$
closes without affecting adiabatic evolution in $\ulbig$. Let $\L^{(s)}$
be a Lindbladian family with a single jump operator
\begin{equation}
F^{(s)}=\begin{pmatrix}0 & \cos^{2}\pi s\\
0 & 2
\end{pmatrix}\,,\label{eq:jump}
\end{equation}
where we have written the matrix in the $|0\ket,|1\ket$ basis ($\ulbig=|0\ket\bra0|$,
$\urbig=|0\ket\bra1|$, etc.). This system has a fixed point $\st=|0\ket\bra0|$,
two coherences $|0\ket\bra1|,|1\ket\bra0|$ with eigenvalue $-\half(4+\cos^{4}\pi s)$,
and one eigenmatrix $\psi$ with presence in $\lrbig$ and eigenvalue
$-\cos^{4}\pi s$. As $s\rightarrow\half$, the gap of $\L_{\lr}^{(s)}$
closes and $\psi\rightarrow|1\ket\bra1|$ becomes a fixed point. However,
the two coherences do not become fixed, as their eigenvalues are $-2$
at $s=\half$. This could not have been possible in the case of a
Hamiltonian (see Fig.~\ref{fig:path}), where coherences between
two degenerate eigenstates are always fixed points. Since no coherences
are there to facilitate transitions between $|0\ket$ and $|1\ket$,
a state in $|0\ket\bra0|$ remains there even at $s=\half$.

\subsection{General case}

The degeneracy in the example below occurs only at a point, and it
has already been shown that point crossings do not affect the adiabatic
limit \cite{Kato1950,Venuti2015}. Here, I show that the adiabatic
evolution of $\ulbig$ is not affected by $\lrbig$ even if the gap
of $\L_{\lr}^{(s)}$ is closed for a \textit{finite segment} of $s$.
Namely, adiabatic evolution (\ref{eq:adiabatic-evolution-A}) of a
state initially in $\ulbig$ is
\begin{equation}
\T(\rout)=\prod_{s\in[0,1]}\R_{\E}^{(s)}(\rout)+O\left(1/N\right)\,,\label{eq:adiafinal}
\end{equation}
where $\R_{\E}^{(s)}$ is the asymptotic projection of the faithful
Lindbladian
\begin{equation}
\L_{\ul}^{(s)}=(\A^{(s)}-\id)_{\ul}=\E^{(s)}-\R_{\ul}^{(s)}\,.\label{eq:efflind}
\end{equation}
This more general statement holds if, for all $s$, $\L_{\of}^{(s)}$
is gapped and the dimension of the fixed-point subspace located in
$\ulbig$ is constant ($\p_{s}\Tr\{\R_{\E}^{(s)}\}=0$). Thus, the
previous result {[}\citealp{ABFJ},~Eq.~(5.16){]} can be extended
to cases where $\L_{\lr}^{(s)}$ also has fixed points. To show this,
we need to consider two cases: case I (II) analyzing any segment of
$s$ during which the gap of $\L_{\lr}^{(s)}$ is closed (open).

For case I, we apply the usual adiabatic theorem for the fixed-point
subspace of $\L^{(s)}$. Therefore, adiabatic evolution is governed
by the Berry connection $A$ (\ref{eq:connect}). Suppressing $s$,
let $\{\psi,j\}$ be any fixed point and its corresponding conserved
quantity of $\L_{\lr}$. Let $\{\St,J\}$ be any such pair for $\L_{\ul}$.
Since we start in a state $\rout\in\ulbig$, we need to consider whether
there are any transitions into $\lrbig$ caused by the terms $\tr\{j^{\dg}\p_{\a}\St\}$
of $A$. We know from Prop.~\ref{prop:3} that $j$ may only overlap
with the fixed-point space of $\L_{\lr}$ and with some part of any
decaying subspace. Since $\ulbig$ is not a decaying subspace, $j=\R_{\lr}(j)$.
Moreover, since the derivative $\p_{\a}$ cannot map $\ulbig$ into
$\lrbig$,\footnote{\label{fn:noleak}This is because of a so-called no-leak condition
\cite{ABFJ}: for any $|\phi_{1}\ket\bra\phi_{2}|\in\ulbig$, $\p_{\a}$
acts on either the ket or the bra part by the product rule and never
on both at the same time. Therefore, $\p_{\a}(|\phi_{1}\ket\bra\phi_{2}|)\in\thubig$.} we have that $\p_{\a}\St=\R_{\thu}(\p_{\a}\St)$. Combining these
yields
\begin{equation}
\tr\{j^{\dg}\p_{\a}\St\}=\tr\left\{ \R_{\lr}(j^{\dg})\R_{\thu}(\p_{\a}\St)\right\} =0\,,
\end{equation}
meaning that $A$ does not contain any transitions of $\ulbig$ fixed
points into those in $\lrbig$. This is true more generally, namely,
for the decomposition of the fixed-point set into blocks indexed by
$\varkappa$ discussed in Sec.~\ref{subsec:Algorithm-for-finding},
each block evolves separately in the adiabatic limit and no transitions
between different $\varkappa$ are allowed.

For case II, $\L_{\lr}$ is always gapped and so adiabatic evolution
is determined by the Berry connection of the fixed points of $\L_{\ul}$.
Therefore, in the adiabatic limit, a state in $\ulbig$ will remain
in the instantaneous $\ulbig$ as $s$ is slowly varied. Moreover,
one can instead consider only $\ps$ because the derivative operator
cannot map $\ulbig$ into $\lrbig$,$^{\ref{fn:noleak}}$
\begin{equation}
\tr\{J^{\dg}\p_{\a}\St\}=\tr\{\R_{\di}(J^{\dg})\R_{\thu}(\p_{\a}\St)\}=\tr\{J_{\ul}^{\dg}\p_{\a}\St\}\,.
\end{equation}

Thus, in both of the above cases, one only had to care about the parts
of the Berry connection $A$ consisting of the fixed points and conserved
quantities of $\L_{\ul}$ (\ref{eq:efflind}). Therefore, the adiabatic
evolution can be expressed using only the associated projection $\ps$
(\ref{eq:adiafinal}).

\section{\label{sec:Application:Matrix-product-state}Application: Matrix
product states}

For those who skimmed Secs. \ref{sec:Structure-of-conserved}-\ref{sec:Adiabaticity-of-Kraus},
those parts focused on the distinction between a channel $\A$ and
its corresponding faithful channel $\E\equiv\R_{\ul}\A\R_{\ul}$ ---
$\A$ restricted to the largest invariant subspace $\ulbig$. The
block $\lrbig$ thus forms a decaying subspace, but the asymptotic
projection $\ppp$ (\ref{eq:ap}) of $\A=\{A^{\ell}\}$ nevertheless
retains information from states in $\lrbig$ by transferring it into
$\ulbig$ through the operators $A_{\ur}^{\ell}$. A similar effect
is observed in the thermodynamic limit of matrix product states (MPS).
Here, the results about $\A$ are applied to obtain an unambiguous
thermodynamic limit for \textit{any} MPS that is translationally invariant
in the bulk, but has non-trivial boundary effects. Then, I show how
one can absorb any dependence of said limit on the decaying parts
$\lrbig$ of the bond degrees of freedom into the boundary conditions
$B$. This allows one to shorten the bond dimension and use the $\A_{\ul}=\E$
instead of the full $\A$. 

\begin{numtheorem}{{Proposition 4.}}

Let $|\Phi_{\A}^{\{B\}}\ket$ be an MPS with transfer channel $\A$
and boundary matrix $B$. Let $O$ be an operator on a single site
of a chain of length $M$. Then, there exists a $\bo$ such that
\begin{equation}
\lim_{M\rightarrow\infty}\bra\P_{\A}^{\{B\}}|O|\P_{\A}^{\{B\}}\ket=\lim_{M\rightarrow\infty}\bra\P_{\E}^{\{\bo\}}|O|\P_{\E}^{\{\bo\}}\ket\,,\label{eq:main-sat}
\end{equation}
where $\E=\A_{\ul}$ is the restriction of $\A$ to the maximal invariant
subspace $\ulbig$.

\end{numtheorem}A similar result holds for two single-site observables
that are infinitely far apart in the chain and infinitely far from
the boundaries. The proof consists of the remainder of this section.
Due to the aforemenetioned effects, $\bo\neq B_{\ul}$; an expression
for $\bo$ in terms of $B$ is derived below. The technique is somewhat
reverse of what has been done before (see Sec. 3.2.2 of \cite{Perez-Garcia2006}):
instead of first considering a general MPS, I simplify the corresponding
MPS of a general faithful channel $\E$ in the thermodynamic limit,
and then show how any contribution from its extension $\A$ can be
absorbed into the boundary.

\subsection{MPS from faithful channels}

Our playground is now a one-dimensional chain of $2M+1$ spins. Each
spin is $L$-dimensional and indexed by the physical index $\ell$.
Let us consider a faithful channel $\E=\{E_{\ell}\}_{\ell=1}^{L}$
(with $E_{\ell}$ $N\times N$ matrices for some \textit{bond dimension}
$N$) and write its corresponding MPS $|\Phi_{\E}^{\{B\}}\ket$:\footnote{Since applying identical transformations $U$ to each site is the
same as changing basis for the Kraus operators of $\E$, $E_{\ell}\rightarrow\sum_{\ell^{\prime}}U_{\ell\ell^{\prime}}E_{\ell^{\pr}}$,
more technically this is a study of sets of MPS related by local unitaries.} 
\begin{equation}
|\P_{\E}^{\{B\}}\ket\propto\!\!\!\!\sum_{\ell_{-M},\cdots,\ell_{M}=0}^{L-1}\!\!\!\!\tr\left\{ BE_{\ell_{-M}}\cdots E_{\ell_{M}}\right\} |\ell_{-M}\cdots\ell_{M}\ket\,,\label{eq:mps}
\end{equation}
where the $N\times N$ boundary matrix $B$ provides the ability to
pin the boundary of the chain to various states \cite{Ueda2011}.
Physically meaningful boundaries are either $B=I$ (the identity)
for translationally invariant MPS's or $B=|r\ket\bra l|$ for some
states $|r\ket,|l\ket$ quantifying the effect of the boundary on
the right and left ends of the chain. I consider cases where the bond
dimension $N$ is independent of system size $M$, noting there are
interesting cases where this is not so {[}\citealp{Perez-Garcia2006},~Appx.~A.1{]}.

The Kraus operators of $\E$ decompose into blocks $\varkappa$, as
in eq.~(\ref{eq:decomp}), with each block a noiseless subsystem.
Since the Kraus operators are block-diagonal, $\tr\{BE_{\ell_{-M}}\cdots E_{\ell_{M}}\}$
decouples into a sum of traces over each block. Each block corresponds
to its own MPS, and the different MPS will not overlap with each other
in the thermodynamic limit ($M\rightarrow\infty$). Thus, we can consider
only one block from now on. In this block, the Kraus operators factor
as $E_{\ell}=U\otimes R_{\ell}$, where $\{R_{\ell}\}$ is an irreducible
channel. The DFS part $U$ can be diagonalized, so the MPS once again
decouples. This leaves us with only the irreducible parts $R_{\ell}$
and automatically puts our MPS into canonical form \cite{Perez-Garcia2006,Cirac}.
We assume from now on that $\E$ is irreducible.

One can obtain the normalization of the MPS in the thermodynamic limit
by tracing out the sites one by one:
\begin{align}
\lim_{M\rightarrow\infty}\bra\P_{\A}^{\{B\}}|\P_{\A}^{\{B\}}\ket & =\lim_{M\rightarrow\infty}\Tr\{\E^{\a(2M+1)}\B\}\\
 & =\Tr\{\ps\B\}\,,
\end{align}
where $\B\equiv B\ot\overline{B}$, the trace is over superoperator
space, and $\a$ is the parameter that eliminates phases stemming
from rotating points. We assume $\E$ has $N$ rotating points, meaning
that all peripheral eigenvalues are $N$th roots of unity. Thus, setting
$\a=N$ yields an unambiguous thermodynamic limit for general boundary
conditions.\footnote{Note that $|\Psi_{\ps}^{\{B\}}\ket$ is also the fixed-point MPS that
$|\Phi_{\E}^{\{B\}}\ket$ flows to under RG transformations \cite{Verstraete2005,Wei2010,Cirac2017},
and $\lim_{M\rightarrow\infty}\bra\P_{\E}^{\{B\}}|\P_{\E}^{\{B\}}\ket=\bra\Psi_{\ps}^{\{B\}}|\Psi_{\ps}^{\{B\}}\ket$,
so simplifying $\ps$ also yields insight into the structure of RG
fixed points.} Similarly, the expectation value of an observable $O$ on a site
in this thermodynamic limit is
\begin{equation}
\lim_{M\rightarrow\infty}\bra\P_{\E}^{\{B\}}|O|\P_{\E}^{\{B\}}\ket=\Tr\{\ps\O\ps\B\}\,,\label{eq:obs}
\end{equation}
where the corresponding superoperator is
\begin{equation}
\O\equiv\sum_{k,\ell=0}^{L-1}\bra\ell|O|k\ket E_{k}\ot\overline{E_{\ell}}\,.\label{eq:O-super}
\end{equation}
Using what we know about $\ps$ (\ref{eq:asymp-irreducible}) for
irreducible channels and plugging in eqs.~(\ref{eq:O-super}) and
(\ref{eq:perm}), eq.~(\ref{eq:obs}) becomes
\begin{align}
\lim_{M\rightarrow\infty}\bra\P_{\E}^{\{B\}}|O|\P_{\E}^{\{B\}}\ket & =N^{2}\sum_{k,\ell=0}^{L-1}\bra\ell|O|k\ket\times\label{eq:exp-1}\\
 & \!\!\!\!\!\!\!\!\!\!\!\!\!\!\!\!\!\!\!\!\!\!\!\!\!\!\!\!\!\!\!\!\!\!\!\!\!\!\!\!\!\!\times\sum_{\m\in\Z_{N}}\tr\left\{ \varPi_{\m}E_{k}\varPi_{\m+1}\varrho E_{\ell}^{\dg}\right\} \tr\left\{ \varPi_{\m+1}B\varPi_{\m}\varrho B^{\dg}\right\} \,.\nonumber 
\end{align}
Each term in the sum over $\m$ is determined by the boundary block
$\varPi_{\m+1}B\varPi_{\m}$, and the unique structure of $\E$ implies
that all other parts of $B$ are irrelevant in the thermodynamic limit.
In fact, each term in the sum corresponds to a contribution coming
from a distinct MPS. Another way to see this is by taking the lattice
size $M$ to be a multiple of $N$ and performing the periodic decomposition
{[}\citealp{Perez-Garcia2006}, Thm.~5{]}.

\subsection{Adding decay}

Now let us extend $\E$ to $\A$, meaning that the MPS $|\P_{\A}^{\{B\}}\ket$
is the same eq.~(\ref{eq:mps}) but with $E_{\ell}\rightarrow A^{\ell}$.
After some algebra, the coefficient $\tr\{B(A^{\ell_{-M}}\cdots A^{\ell_{M}})\}$
becomes equal to
\begin{align}
 & \,\,\,\phantom{+}\tr\left\{ B_{\ul}(E_{\ell_{-M}}\cdots E_{\ell_{M}})\right\} +\tr\left\{ B_{\lr}(A_{\lr}^{\ell_{-M}}\cdots A_{\lr}^{\ell_{M}})\right\} \nonumber \\
 & +{\displaystyle \sum_{m=-M}^{M}}\tr\left\{ B_{\ll}(E_{\ell_{-M}}\cdots E_{\ell_{m-1}})A_{\ur}^{\ell_{m}}(A_{\lr}^{\ell_{m+1}}\cdots A_{\lr}^{\ell_{M}})\right\} \,.\label{eq:mps2}
\end{align}
This is precisely the same structure occurring in, e.g., product vacua
with boundary states \cite{Bachmann2012,Bachmann2014}. The first
term corresponds to the MPS $|\P_{\E}^{\{B\}}\ket$ we already studied.
The second term vanishes in the thermodynamic limit because its corresponding
transfer matrix does not have any fixed points. When $B_{\ll}\neq0$,
the third term is present and has the form of a translationally-invariant
domain wall excitation. Therefore, the decaying subspace $\lrbig$
corresponds to extra degrees of freedom on each site which house such
an excitation. This excitation is never present for periodic boundary
conditions ($B=I$), allowing one to straightforwardly derive an standard
form for that case \cite{Perez-Garcia2006,Cirac}. However, here we
focus on ``twisted'' boundaries $B_{\ll}\neq0$.

The single-site expectation value (\ref{eq:obs}) is now
\begin{equation}
\lim_{M\rightarrow\infty}\bra\P_{\A}^{\{B\}}|O|\P_{\A}^{\{B\}}\ket=\Tr\{\ps\O\ppp\B\}\,,\label{eq:obs-1}
\end{equation}
where the second projection is now $\ppp$ due to the presence of
the boundary term $\B_{\ll}$. Simplifying, the analogue of eq.~(\ref{eq:exp-1})
is then
\begin{align}
\lim_{M\rightarrow\infty}\bra\P_{\A}^{\{B\}}|O|\P_{\A}^{\{B\}}\ket & =N^{2}\sum_{k,\ell=0}^{L-1}\bra\ell|O|k\ket\times\label{eq:exp-2}\\
 & \!\!\!\!\!\!\!\!\!\!\!\!\!\!\!\!\!\!\!\!\!\!\!\!\!\!\!\!\!\!\!\!\!\!\!\!\!\!\!\!\!\!\!\!\!\!\!\!\!\!\!\!\times\sum_{\m\in\Z_{N}}\tr\left\{ \varPi_{\m}E_{k}\varPi_{\m+1}\varrho E_{\ell}^{\dg}\right\} \tr\left\{ \varPi_{\m+1}\ppp(B\varPi_{\m}\varrho B^{\dg})\right\} \,.\nonumber 
\end{align}
Now there are two contributions from $\ppp$, $\ps$ and $\ppp\R_{\lr}$,
with the latter determined by Prop.~\ref{prop:3}. However, it is
possible to absorb both into a new boundary matrix $\bo$ such that
the expectation value in $|\Phi_{\E}^{\{\widetilde{B}\}}\ket$ is
the same as that of $|\P_{\A}^{\{B\}}\ket$. To determine $\bo$,
note that for eqs.~(\ref{eq:exp-1}) (with $B\rightarrow\bo$) and
(\ref{eq:exp-2}) to be equal, we must have
\begin{equation}
\tr\{\varPi_{\m+1}\bo\varPi_{\m}\varrho\bo^{\dg}\}=\tr\left\{ \varPi_{\m+1}\ppp(B\varPi_{\m}\varrho B^{\dg})\right\} \equiv x_{\m}\,.\label{eq:need}
\end{equation}
This $x_{\m}\geq0$ since $B\varPi_{\m}\varrho B^{\dg}$ is positive,
$\ppp$ is a channel, and $\varPi_{\m+1}$ is a projection. Now, write
out the projections, $\varPi_{\m}=\sum_{\xi}|\m,\xi\ket\bra\m,\xi|$
(with $\xi$ depending on $\m$), to obtain
\begin{equation}
\tr\{\varPi_{\m+1}\bo\varPi_{\m}\varrho\bo^{\dg}\}=\sum_{\xi}\varrho_{\m,\xi}\sum_{\xi^{\pr}}|\bra\m+1,\xi^{\pr}|\bo|\m,\xi\ket|^{2}\,,
\end{equation}
where $\varrho_{\m,\xi}\equiv\bra\m,\xi|\varrho|\m,\xi\ket>0$ (since
$\varrho$ is full-rank). Therefore, setting
\begin{equation}
\bra\m+1,\xi|\bo|\m,\xi^{\pr}\ket=\sqrt{N\frac{x_{\m}}{\tr\left\{ \varPi_{\m+1}\right\} }}
\end{equation}
satisfies the equality (\ref{eq:need}). (We used $\tr\{\varPi_{n}\varrho\}=\nicefrac{1}{N}$;
see Sec.~\ref{subsec:Irreducible-channels}.) Thus, we have shown
how to construct a $\bo$ that satisfies eq.~(\ref{eq:main-sat}).

\subsection{Other observables}

The same result occurs with two observables $O^{(1)}$ and $O^{(2)}$
(with corresponding superoperators $\O^{(1)}$ and $\O^{(2)}$) separated
by some number of sites $W$,
\begin{equation}
\!\!\!\!\lim_{M\rightarrow\infty}\bra\P_{\A}^{\{B\}}|O^{(1)}O^{(2)}|\P_{\A}^{\{B\}}\ket=\tr\left\{ \R_{\A}\O^{(1)}\A^{W}\O^{(2)}\ppp\B\right\} ,
\end{equation}
and take the $W\rightarrow\infty$ limit by blocking sites in order
to get rid of any phases from rotating points. This yields
\begin{equation}
\!\!\!\!\lim_{M,W\rightarrow\infty}\bra\P_{\A}^{\{B\}}|O^{(1)}O^{(2)}|\P_{\A}^{\{B\}}\ket=\tr\left\{ \ps\O^{(1)}\ps\O^{(2)}\ppp\B\right\} .
\end{equation}
After the simplifications of the previous section, we observe that
in order for $\Phi_{\E}^{\{\widetilde{B}\}}$ to have the same expectation
values, one now has to tune the ``next nearest-neighbor'' elements
$\tr\{\varPi_{\m+2}\bo\varPi_{\m}\varrho\bo^{\dg}\}$ {[}independent
from $x_{\m}$ (\ref{eq:need}) for $N\geq3${]}. The potentially
interesting case of $N=2$ is left for future investigation.

Similarly, consider an observable touching the left boundary:
\begin{align}
\lim_{M\rightarrow\infty}\bra\P_{\A}^{\{B\}}|O^{(L)}|\P_{\A}^{\{B\}}\ket & =\Tr\{\O^{(L)}\R_{\A}\B\}\,.\label{eq:leftside}
\end{align}
Somewhat surprisingly, considering an observable touching the right
boundary produces something completely different:
\begin{equation}
\lim_{M\rightarrow\infty}\bra\P_{\A}^{\{B\}}|O^{(R)}|\P_{\A}^{\{B\}}\ket=\Tr\{\R_{\A}\O^{(R)}\B\}\,.
\end{equation}
Notice how $\R_{\A}$ now comes \textit{before} the observable {[}cf.
the first equality of Eq. (\ref{eq:leftside}){]}, which results in
a series of new terms stemming from combinations of $A_{\ur}^{\ell}$
and $A_{\lr}^{\ell}$ with $B$. Why is there an asymmetry between
the two boundaries? This has to do with the fact that we had initially
assumed an asymmetric form for our MPS, $A^{\ell}=A_{\thd}^{\ell}$.
The domain wall-type excitations represented by the third term in
Eq. (\ref{eq:mps2}) are such that there is always a $A_{\lr}$ at
the right-most site $M$.

\section{Conclusion}

An important property of quantum channels $\A$ is their asymptotics,
i.e., their behavior in the limit of infinite applications, akin to
the infinite-time limit of Lindbladians \cite{pub011,ABFJ}. An infinite
product of $\A$ produces the channel's asymptotic projection $\ppp$
--- a projection on all of the non-decaying eigenspaces of the channel
(i.e., whose eigenvalues have unit modulus). The superoperator $\ppp$
can be constructed out of the channel's left and right rotating points,
or as they are called here, conserved quantities $J$ and steady-state
basis elements $\St$. The aim of the first half of this work is to
determine which parts of an initial state are conserved in the limit
of infinite applications of the channel, extending analogously motivated
work for Lindbladian channels \cite{pub011}. This involves a derivation
of the structure of both the $\st$ (already known) and the $J$ (developed
here) making up $\ppp$.

I start off with two statements about channels admitting a full-rank
fixed point, which I call faithful. The first is that any $J$ commute
with a faithful channel's Kraus operators up to a phase. The second
is that the eigenvalue of any diagonalizable $J$ of a faithful channel
is an $N$th root of unity, where $N$ is tightly bounded by trace
norm of $J$. A third result deals with determining the dependence
of the asymptotic state on the initial state and on properties of
$\A$. An analytical formula is derived that quantifies the dependence
of the final state on initial states located in $\A$'s decaying eigenspaces
(i.e., whose eigenvalues are less than one in modulus).

The aim of the second half of this work is to apply results from the
first half to an adiabatic limit for channels and to matrix product
states. In both applications, it is shown that the channel $\E$---the
restriction of $\A$ to its largest invariant subspace---is sufficient
to work with for the outlined purposes. An adiabatic limit for channels
in developed, and it is shown that the gap of the part of the channel
acting on the decaying subspace can close. The second application
is to matrix product states (MPS), where asymptotics come into play
in the thermodynamic limit or in the limit of infinite renormalization
transformations. In the same way that asymptotic states depend on
initial states, the thermodynamic limit of MPS (whose transfer matrices
admit more than one fixed point) depends on the boundary conditions.
In such situations, the effects of any decaying bond degrees of freedom
can be absorbed in the boundary conditions. Quantitatively, it is
shown that the thermodynamic expectation value of a local operator
$O$ with an MPS having transfer matrix $\A$ and boundary condition
$B$ is equivalent to the expectations values with MPS having transfer
matrix $\E$ and a modified boundary condition $\bo$ (\ref{eq:main-sat}).
This allows one to remove extra bond degrees of freedom when considering
the thermodynamic limit. Since similar two-dimensional MPS (often
called ``PEPS'' \cite{Cirac2011}) also correspond to a transfer
channel, such techniques may be further generalized to study PEPS
dependence on boundaries.
\begin{acknowledgments}
Insightful discussions with B. Bradlyn, X. Chen, M. Fraas, L. Jiang,
D. Perez-Garcia, M. B. Sahinoglu, N. Schuch, F. Ticozzi, A. M. Turner,
and M. M. Wolf are acknowledged. This research was supported in part
by the National Science Foundation (PHY17-48958) and the Walter Burke
Institute for Theoretical Physics at Caltech. I thank KITP Santa Barbara
for their hospitality as part of the Quantum Physics of Information
workshop.
\end{acknowledgments}

\appendix

\section{Proofs\label{sec:Proofs}}
\begin{prop}
\label{prop:1}Let $\E=\left\{ E_{\ell}\right\} $ be a faithful channel.
Let $J$ be a conserved quantity of $\E$, i.e., $\E^{\dgt}\left(J\right)=e^{-i\la}J$
for some real $\la$. Then, for all $\ell$,
\begin{equation}
JE_{\ell}=e^{-i\la}E_{\ell}J\,.
\end{equation}
\end{prop}

\begin{proof}
This result reduces to known results for irreducible \cite[Thm. 4.2]{Evans1978},
ergodic \cite[Thm. 9]{Burgarth2013a}, or unital \cite[Thm. 4]{Bialonczyk2017}
channels. Another proof is using Thms. 4.1-4.2 and Corollary 4.3 in
Ref. \cite{Novotny2012}. This proof extends an often-used \cite{robin,Choi2006}
application of the dissipation function \cite{Lindblad1976} from
fixed points to rotating points. An analogous extension for Lindbladians
is in Ref. \cite{ABFJ}. Let 
\begin{equation}
X_{\ell}\equiv JE_{\ell}-e^{-i\la}E_{\ell}J
\end{equation}
for each Kraus operator $E_{\ell}$. Then, after some algebra, 
\begin{equation}
\sum_{\ell}X_{\ell}^{\dg}X_{\ell}=\E^{\dgt}\left(J^{\dg}J\right)-J^{\dg}J\,.
\end{equation}
Now multiply both sides by a full-rank fixed point $\rss$ and take
the trace. Moving $\E^{\dgt}$ under the Hilbert-Schmidt inner product
so that it acts on $\rss$ yields
\begin{equation}
\tr\left\{ \E\left(\rss\right)J^{\dg}J\right\} -\tr\left\{ \rss J^{\dg}J\right\} =0
\end{equation}
for the right-hand side, meaning that
\[
\sum_{\ell}\tr\left\{ \rss X_{\ell}^{\dg}X_{\ell}\right\} =0\,.
\]
Since $X_{\ell}^{\dg}X_{\ell}\geq0$ and $\rss>0$, the only way for
the above to hold is for $X_{\ell}=0$, which implies the statement.
\end{proof}
\begin{prop}
\label{prop:2}Let $\E=\left\{ E_{\ell}\right\} $ be a faithful channel.
Let $\jd$ be such that $\E^{\dgt}(\jd)=e^{-i\la}\jd$ for some real
$\la$ and assume $\jd$ is diagonalizable. Then, there exists an
integer $n$ such that 
\begin{equation}
\la=\frac{2\pi}{N}n\,\,\,\,\,\,\,\text{for some}\,\,\,\,\,\,\,N\leq\left\Vert \jd\right\Vert _{1}\,,
\end{equation}
where $\left\Vert \,\,\,\right\Vert _{1}$ is the trace norm.
\end{prop}

\begin{proof}
The tighter bound compelements similar work (\cite{groh}; \cite{wolf2010},
Thm. 6.6; \cite{Fannes1992}, Prop. 3.3; \cite{Bialonczyk2017}, Corr.
3; \cite{baumr}, Prop. 28). First, there must exist an $N\geq1$
such that 
\begin{equation}
\varPi\equiv\jd^{N}\label{eq:proj}
\end{equation}
is a projection ($\jd^{2N}=\jd^{N}$). To show this, assume by contradiction
that all powers of $\jd$ are distinct. Then, there is an infinite
sequence of conserved quantities $\jd^{N}$ with eigenvalues $e^{-iN\la}$
due to Prop.~\ref{prop:1}. But the Hilbert space is $D$-dimensional,
so there are at most $D^{2}$ fixed/rotating points. Moreover, $e^{-iN\la}=1$;
otherwise, $\jd^{N}$ would have a different eigenvalue than $\jd^{2N}$.
Therefore, there exists an $N\leq D^{2}$ such that $\la=\frac{2\pi}{N}n$
for $n\in\left\{ 0,1,\cdots N-1\right\} $, i.e., $\jd$ has eigenvalues
which are $N$th roots of unity.

Now I show that $N\leq\left\Vert \jd\right\Vert _{1}$. Since $\jd$
is diagonalizable, one can write
\begin{equation}
\jd=\sum_{n\in\Z_{N}}\o^{n}\varPi_{n}\,,
\end{equation}
where $\o\equiv e^{i\frac{2\pi}{N}}$ and the projections onto the
eigenspace of $\jd$ corresponding to eigenvalue $\o^{n}$ are
\begin{equation}
\varPi_{n}=\frac{1}{N}\sum_{k\in\Z_{N}}\o^{-nk}\jd^{k}\,.
\end{equation}
Proposition~\ref{prop:1} implies that for all $\ell$,
\begin{equation}
\varPi_{n}E_{\ell}=E_{\ell}\varPi_{n+1}\,\,\,\,\,\,\,\,\text{modulo }N\,.\label{eq:perm}
\end{equation}
Projecting $E_{\ell}$ onto the support of $\jd$ using $\varPi=\sum_{n}\varPi_{n}$
(\ref{eq:proj}) and applying the above yields a matrix of exactly
$N$ blocks $\varPi_{n}E_{\ell}\varPi_{n+1}$,
\begin{equation}
\varPi E_{\ell}\varPi=\sum_{n\in\Z_{N}}\varPi_{n}E_{\ell}\varPi_{n+1}\,.
\end{equation}
In the basis of eigenstates of $\varPi$, this can be viewed as a
matrix of dimension
\begin{equation}
\tr\varPi=\tr\jd^{N}=\left\Vert \jd\right\Vert _{1}\,.
\end{equation}
Since $\varPi_{n}$ are distinct orthogonal projections, the largest
number of blocks in $\varPi E_{\ell}\varPi$ occurs when all $\varPi_{n}$
are rank one; then $N=\left\Vert \jd\right\Vert _{1}$. Generally,
$N\leq\left\Vert \jd\right\Vert _{1}$.
\end{proof}
\begin{prop}
\label{prop:3}The conserved quantities of $\A$ corresponding to
eigenvalues $e^{i\la}$ areThe conserved quantities of $\A$ corresponding
to eigenvalues $e^{i\la}$ are
\begin{equation}
J=J_{\ul}+J_{\lr}=J_{\ul}+\RR_{\lr}^{(\la)\dgt}\A^{\dgt}(J_{\ul})\,,
\end{equation}
where $J_{\ul}$ are the conserved quantities of $\A_{\ul}=\E$ and
\begin{equation}
\RR_{\lr}^{(\la)}=-\left(\A-e^{i\la}\right)_{\lr}^{-1}\,.\label{eq:resolvent-1}
\end{equation}
\end{prop}

\begin{proof}
I generalize previous results (\cite{robin}, Lemma 5.8; \cite{Cirillo2015},
Prop. 7) to the case of rotating points. Start by writing the eigenvalue
equation 
\begin{equation}
Je^{-i\la}=\A^{\dgt}\left(J\right)
\end{equation}
in terms of the four-corners decomposition of $J$ and 
\begin{equation}
\A^{\dgt}=\left[\begin{array}{ccc}
\E_{\ul}^{\dgt} & 0 & 0\\
\R_{\of}\A^{\dgt}\R_{\ul} & \A_{\of}^{\dgt} & 0\\
\R_{\lr}\A^{\dgt}\R_{\ul} & \R_{\lr}\A^{\dgt}\R_{\of} & \A_{\lr}^{\dgt}
\end{array}\right]\,.
\end{equation}
The three zeroes in the above decomposition for $\A^{\dgt}$ can be
derived by brute-force use of eq. (\ref{eq:bl}). The eigenvalue equation
is equivalent to\begin{subequations} 
\begin{align}
J_{\ul}e^{-i\la} & =\A_{\ul}^{\dg}(J_{\ul})\\
J_{\of}e^{-i\la} & =\R_{\of}\A^{\dgt}(J_{\ul})+\A_{\of}^{\dgt}(J_{\of})\label{eq:of}\\
J_{\lr}e^{-i\la} & =\R_{\lr}\A^{\dgt}(J_{\ul})+\R_{\lr}\A^{\dgt}(J_{\of})+\A_{\lr}^{\dgt}(J_{\lr})\,.\label{eq:lr}
\end{align}
\end{subequations}First, let's look at the term (\ref{eq:of}):
\begin{equation}
\R_{\of}\A^{\dgt}(J_{\ul})=\R_{\ur}\A^{\dgt}(J_{\ul})+\R_{\ll}\A^{\dgt}(J_{\ul})\,.
\end{equation}
Using eqs. (\ref{eq:com}) and (\ref{eq:conds}), one can see that
\begin{equation}
\R_{\ur}\A^{\dgt}(J_{\ul})=\sum_{\ell}A_{\ul}^{\ell\dg}J_{\ul}A_{\ur}^{\ell}=e^{-i\la}J_{\ul}\sum_{\ell}A_{\ul}^{\ell\dg}A_{\ur}^{\ell}=0\,,
\end{equation}
and similarly for the second term $\R_{\ll}\A^{\dgt}(J_{\ul})$. This
reduces eq. (\ref{eq:of}) to
\begin{equation}
J_{\of}e^{-i\la}=\A_{\of}^{\dgt}(J_{\of})\,.
\end{equation}
This can in turn be used to show that $J_{\of}=0$ \cite{ABFJ}. Assume
by contradiction that $J_{\of}\neq0$. Then, there must exist a corresponding
right fixed point $\St$ with $\A_{\of}(\St)=0$. But we have already
assumed that all right fixed points are in $\ulbig$. Therefore, $J_{\of}=0$.

The remaining eq. (\ref{eq:lr}) becomes
\begin{equation}
J_{\lr}e^{-i\la}=\R_{\lr}\A^{\dgt}(J_{\ul})+\A_{\lr}^{\dgt}(J_{\lr})
\end{equation}
and can be used to solve for $J_{\lr}$ in terms of $J_{\ul}$ (since
$\A-e^{i\la}$ is invertible on $\lrbig$), obtaining the statement.
\end{proof}
\bibliographystyle{unsrtnat}
\bibliography{C:/Users/russi/Documents/library}

\end{document}